\documentclass[10pt,aps,pra,onecolumn,superscriptaddress,floatfix,nofootinbib,amsmath,amsfonts,amssymb]{revtex4-2}

\usepackage{amsthm}
\usepackage{braket}
\usepackage{graphicx}
\usepackage{xcolor}
\usepackage{fullpage}
\usepackage{microtype}
\usepackage{booktabs}
\usepackage{hyperref}
\usepackage[capitalise,nameinlink]{cleveref}
\usepackage{float}

\definecolor{red4}{Hsb}{0,0.79,0.72}
\definecolor{blue4}{Hsb}{240,0.65,0.85}
\definecolor{purple4}{Hsb}{330,0.79,0.70}
\hypersetup{
    colorlinks,
    citecolor=blue4,
    linkcolor=red4,
    urlcolor=purple4,
    breaklinks=true,
}

\newtheorem{theorem}{Theorem}

\newtheorem{lemma}[theorem]{Lemma}

\newtheorem{definition}{Definition}
\newtheorem*{remark}{Remark}
\newtheorem{manualtheoreminner}{Theorem}
\newtheorem{manuallemmainner}{Lemma}
\newenvironment{manualtheorem}[1]{%
  \IfBlankTF{#1}
    {\renewcommand{\themanualtheoreminner}{\unskip}}
    {\renewcommand\themanualtheoreminner{#1}}%
  \manualtheoreminner
}{\endmanualtheoreminner}
\newenvironment{manuallemma}[1]{%
  \IfBlankTF{#1}
    {\renewcommand{\themanuallemmainner}{\unskip}}
    {\renewcommand\themanuallemmainner{#1}}%
  \manuallemmainner
}{\endmanuallemmainner}

\DeclareMathOperator{\diag}{diag}
\DeclareMathOperator{\re}{Re}

\DeclareMathOperator{\erf}{erf}

\crefname{section}{Sec.}{Secs.}

\newcommand{\ketbra}[2]{\left\vert{#1}\middle\rangle\middle\langle{#2}\right\vert}

\makeatletter
\newcommand{\apptocfile}{atoc}
\let\apptoc@orig@appendix\appendix
\renewcommand{\appendix}{%
  \apptoc@orig@appendix
  \let\apptoc@orig@addtocontents\addtocontents
  \long\def\addtocontents##1##2{%
    \def\apptoc@ext{##1}%
    \def\apptoc@toc{toc}%
    \ifx\apptoc@ext\apptoc@toc
      \apptoc@orig@addtocontents{\apptocfile}{##2}%
    \else
      \apptoc@orig@addtocontents{##1}{##2}%
    \fi
  }%
}
\newcommand{\appendixtableofcontents}{%
  \begingroup
    \setcounter{tocdepth}{3}%
    \phantomsection
    \let\addcontentsline\@gobblethree
    \section*{Contents}%
    \pdfbookmark[1]{Appendices}{apxcontents}%
    \@starttoc{\apptocfile}%
  \endgroup
}
\makeatother

\allowdisplaybreaks

\begin{document}
\title{\texorpdfstring{QKAN: Quantum Kolmogorov-Arnold Networks with Applications \\in Machine Learning and Multivariate State Preparation}{QKAN: Quantum Kolmogorov-Arnold Networks with Applications in Machine Learning and Multivariate State Preparation}}

\author{Petr Ivashkov}
\email{pivashkov@ethz.ch}
\affiliation{Centre for Quantum Technologies, National University of Singapore, Singapore}
\affiliation{Department of Information Technology and Electrical Engineering, ETH Zürich, Zürich, Switzerland}
\author{Po-Wei Huang}
\affiliation{Centre for Quantum Technologies, National University of Singapore, Singapore} 
\affiliation{Mathematical Institute, University of Oxford, Oxford, United Kingdom} 
\author{Kelvin Koor}
\affiliation{Centre for Quantum Technologies, National University of Singapore, Singapore}
\author{Lirandë Pira}
\email{lpira@nus.edu.sg}
\affiliation{Centre for Quantum Technologies, National University of Singapore, Singapore}
\author{Patrick Rebentrost}
\email{patrick@comp.nus.edu.sg}
\affiliation{Centre for Quantum Technologies, National University of Singapore, Singapore}
\affiliation{Department of Computer Science, National University of Singapore, Singapore}

\begin{abstract}
We introduce quantum Kolmogorov-Arnold networks (QKAN), a quantum algorithmic framework inspired by the recently proposed Kolmogorov-Arnold Networks (KAN). QKAN inherits the compositional structure of KAN and is based on block-encodings, constructed recursively from a single layer using quantum singular value transformation. We demonstrate the algorithmic utility of QKAN in two applications. First, we introduce and analyze QKAN as a quantum learning model, treating the eigenvalues of block-encoded matrices as neurons and applying parametrized activation functions on the edges of the network. We show that QKAN is a wide-and-shallow neural architecture, where shallow depth is compensated by exponentially wide layers whenever efficient block-encodings of inputs are available. We further discuss how to parametrize and train QKAN using parametrized quantum circuits and quantum linear algebra subroutines. Second, we demonstrate that QKAN can serve as a multivariate quantum state-preparation protocol for functions with shallow compositional structure. We demonstrate this by efficiently preparing a multivariate Gaussian quantum state using a two-layer QKAN. Looking forward, we anticipate that QKAN’s compositional and modular design will enable new applications in quantum machine learning and quantum state preparation.
\end{abstract}

\maketitle

\section{Introduction} \label{sec:introduction}

Kolmogorov-Arnold representation theorem (KART) states that any continuous function of multiple variables can be decomposed using two layers of composition and summation of univariate functions~\citep{kolmogorov_representation_1956, kolmogorov_representation_1957,arnold_functions_1957, arnold_representation_1959}. Recently,  \citet{liu2025kan} extended this compositional structure beyond two layers, providing an alternative neural network design aimed at offering advantages over traditional feedforward multilayer perceptrons (MLPs)~\citep{krizhevsky_imagenet_2012, goodfellow_deep_2016, lecun_deep_2015, he2015deepresiduallearningimage}. Although KANs do not inherit the universal representation property of KART, their structure, based on compositions of parametrized univariate activation functions, can yield better interpretability and improved accuracy on small-scale tasks~\citep{liu2025kan}. In scientific applications, where many target functions admit symbolic formulas, KANs can reveal modular structure and potentially aid in the discovery of new physical laws, making them a promising tool for scientific discovery~\citep{liu2024kan20}. The KAN architecture has inspired multiple extensions and applications, including Convolutional KANs~\citep{bodner_convolutionalkolmogorov_2024}, Graph KANs~\citep{kiamari2024GKAN, decarlo2024KAGN}, Chebyshev KANs~\citep{ss2024chebyshevpolynomialbasedkolmogorovarnoldnetworks}, KANs for quantum circuits~\citep{kundu2024kanqaskolmogorovarnoldnetworkquantum} and others~\citep{bozorgasl2024wavKAN,genet2024KANtransformer, wang2024KANinformed, aghaei2024rationalkan, xu2024fourierKANGCF, aghaei2024fKAN, zhang2024graphKAN}.

In this work, motivated by the potential of KANs in the classical setting, we introduce a quantum version, QKAN, a structured quantum architecture that leverages the quantum singular value transformation (QSVT) to apply nonlinear transformations. QSVT applies polynomial transformations to the singular values of a matrix encoded as a block of a unitary (a \textit{block-encoding}), utilizing the power of quantum computers to manipulate exponentially large unitary operators efficiently~\citep{gilyen2019quantum}. QSVT has seen widespread adoption as a quantum meta-algorithm, both rederiving previous algorithms~\citep{rall2023amplitude, martyn2021grand} and designing new efficient quantum algorithms~\citep{gilyen2022improved,wang2023quantum,wang2024new,li2024exponential,qiu2024quantum}. QKAN uses block-encodings as its input and output model, representing both the input and output vectors as block-encoded diagonal operators, which can be manipulated using quantum linear algebra subroutines. We demonstrate the algorithmic utility of QKAN in two applications.

First, we introduce and analyze QKAN as a quantum learning model. In quantum machine learning, models are developed using quantum mechanical principles~\citep{schuld_machine_2021, biamonte_quantum_2017, beer_training_2020}. Existing approaches include variational quantum algorithms (VQAs) employing parametrized quantum circuits whose parameters are optimized to minimize a cost function~\citep{cerezo2021variational, havlicek2019supervised,mitarai2018quantum, farhi2018classification, benedetti2019parameterized, schuld2019quantum}, similar to MLPs. Their generalization, expressibility, and interpretability have been extensively studied in Refs.~\citep{du2020expressive, abbas2021power, schuld2021effect, banchi_generalization_2021,holmes_connecting_2022, caro2022generalization, pira_interpretabilityquantumneuralnetworks_2024}. In the fault-tolerant regime, various quantum implementations of classical machine learning algorithms have been proposed, including support vector machines~\citep{rebentrost_quantum_2014}, deep convolutional neural networks~\citep{kerenidis2020qcnn}, transformers~\citep{guo_quantum_2024}, and various others~\citep{allcock2020quantum, rebentrost2018quantum, amin2018quantum, kapoor2016quantum}. Contrary to previous architectures, QKAN treats the eigenvalues of block-encoded matrices as neurons and applies parametrized activation functions on network edges via linear combinations of Chebyshev polynomials, or other basis functions that can be realized efficiently using QSVT. The gate complexity of QKAN scales linearly with the cost of constructing the block-encoding of an $N$-dimensional input vector, which in certain cases, such as for inherently quantum inputs, can be $\mathcal{O}(\mathrm{polylog}(N))$. At the same time, composing layers incurs an exponential overhead in depth due to the recursive QSVT-based construction, so QKAN is naturally constrained to be shallow. This makes QKAN a wide-and-shallow architecture: when efficient block-encodings are available, a shallow QKAN can realize exponentially wide layers at a polylogarithmic cost, a regime that is inaccessible to classical neural networks. For example, given access to a quantum unitary that prepares a $N$-dimensional quantum state of interest efficiently, we can process that state by computing a multivariate function of its amplitudes in $\mathcal{O}(\mathrm{polylog}(N))$ time, assuming that the target function admits an efficient polynomial approximation. Such an operation generally requires $\mathcal{O}(N)$ classical runtime. We note that although we implement QKAN with Chebyshev polynomials to facilitate training and interpretability, QKAN is not restricted to the Chebyshev basis and can employ any bounded-degree, bounded-range polynomials realizable via QSVT.

Second, we demonstrate that QKAN can serve as a multivariate quantum state-preparation protocol. The goal of quantum state preparation is to prepare a quantum state, for example, for use in other quantum algorithms. The problem of loading univariate functions has been extensively investigated in the prior literature~\citep{grover2000synthesis,plesch2011quantum,sanders2019blackbox,zhang2022quantum,mcardle2022quantum,rattew2023nonlineartransformationsquantumamplitudes,gonzalezconde2024efficient}. However, extensions to multivariate state preparation remain scarce despite their importance~\citep{rosenkranz2025quantum, bauer2021practical, manabe2025state}. We illustrate how QKAN's compositional circuitry can efficiently prepare families of multivariate high-dimensional distributions by exploiting their compositional structure. As a concrete example, we show that a two-layer QKAN can efficiently load a $D$-dimensional Gaussian distribution into a quantum state. 

\begin{figure}[tb]
    \centering
    \includegraphics[width=\textwidth]{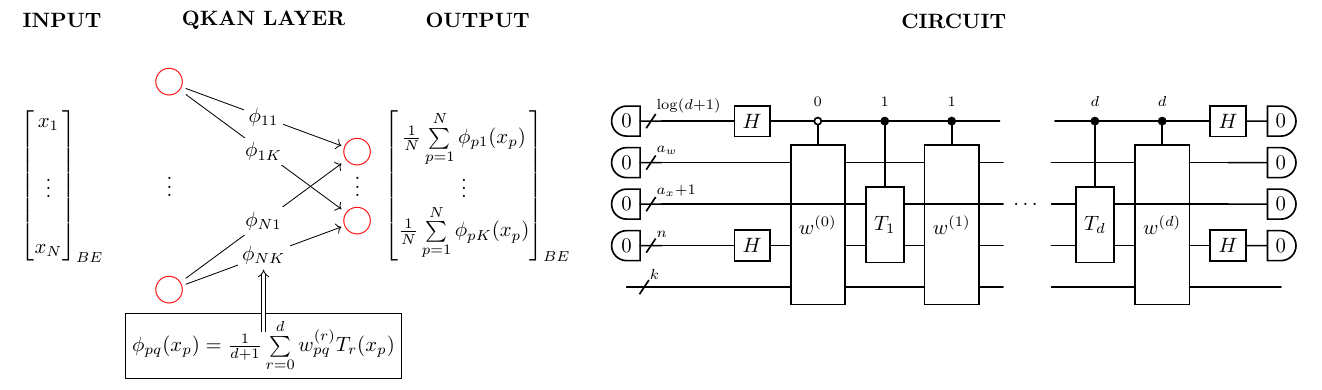}
    \caption{\textit{Construction of a CHEB-QKAN layer with the corresponding quantum circuit.} The input to the QKAN model is a diagonal block-encoding of an $N$-dimensional real vector $\vec{x}$. The CHEB-QKAN layer applies univariate activation functions $\phi_{pq}$ to each input component $x_p$, where $p\in[N]$ indexes input nodes and $q\in[K]$ indexes output nodes. The output vector is computed as a sum over activated input nodes. This operation yields a block-encoded real $K$-dimensional output vector. The quantum circuit implementation requires $1+\log_2(d+1)$ qubits for the construction and linear combination of weighted Chebyshev polynomials, $a_w + a_x$ qubits for the block-encodings of input and weights, $n=\log_2 N$ qubits for input vector encoding, and $k=\log_2 K$ qubits for output. The circuit consists of a series of multi-controlled block-encodings of Chebyshev polynomials, interspersed with diagonal block-encodings of the corresponding real weights. The entire circuit represents a block-encoding of the $K$-dimensional vector corresponding to the CHEB-QKAN layer, with auxiliary qubits initialized and measured in the $\ket{0}$ state. }
    \label{fig:MAIN}
\end{figure}

\section{Results}\label{sec:results}

\subsection{Notation and Preliminaries}\label{sec:preliminaries_and_notation}
Throughout this manuscript, $N$ denotes the dimension of input, assumed to be a power of two, and $n = \log_2(N)$ represents the number of qubits. $K$ denotes the dimension of the output and, similarly, $k = \log_2(K)$. The subscript $n$ in $\ket{\psi}_n$ and $U_n$ indicates the size of the system in terms of qubits. For a vector $v$, $\|v\|_p$ is the $\ell_p$-norm of $v$. For a matrix $A$, $\|A\|$ is the spectral norm of $A$. Further, $\diag(x_1, x_2, \dots, x_N)$ represents a diagonal matrix whose diagonal entries are $x_1, x_2, \dots, x_N$. Here, $T_r(x) \in \mathbb{R}[x]$ is the $r$-th Chebyshev polynomial of the first kind, defined as $T_r(x) := \cos(r \arccos(x))$. We denote $\mathbb{R}$[x] as the set of all polynomials with real coefficients in the variable $x$. We adopt the convention that indices in summations run from 1 to the upper limit of the sum and use $[N]$ to denote the set $\{1,\dots,N\}$.

\subsection{Kolmogorov-Arnold Networks (KAN)}\label{sec:KAN}
Kolmogorov-Arnold representation theorem states that any continuous multivariate function can be represented as a composition of univariate functions with the summation~\citep{kolmogorov_representation_1956, arnold_functions_1957, braun_constructive_2009}. Formally, for any continuous function $f: [0,1]^N \to \mathbb R$, there exist continuous inner functions $\phi_{pq}:[0,1]\to \mathbb R$ (independent of $f$) and outer functions $g_q : \mathbb R \to \mathbb R$ (dependent of $f$) such that
\begin{equation}
    f(x_1,\cdots,x_N) = \sum_{q=1}^{2N+1} g_q\left (\sum_{p=1}^N \phi_{pq}\left(x_p\right)\right).
\end{equation}

In the context of neural networks, KART has been studied to explain how deep learning can overcome the curse of dimensionality, with one approach involving the approximation of the inner and outer functions of the KART representation using neural networks. However, the practicality of this approach is limited by the high nonsmoothness of the inner and outer functions, even when the original function is smooth, posing significant challenges in accurate approximation and robustness to noise.

\citet{liu2025kan} proposed generalizing the compositional structure of KART to include more layers. This architecture, called Kolmogorov-Arnold Network, can contain an arbitrary number of layers, as opposed to the two layers guaranteed in KART. Here, we define KAN as outlined in the original study and only slightly adapt the notation of Ref.~\citep{liu2025kan}.
\begin{definition}[KAN layer,~\citep{liu2025kan}]
    Define a KAN layer as a transformation $\Phi: \mathbb{R}^N \rightarrow \mathbb{R}^K$ that takes a real vector $\vec{x}$ as input and outputs a real vector $\Phi(\vec{x})$ such that
	\begin{equation}
		\Phi(\vec{x}) = \left(\sum_{p=1}^N \phi_{p1}(x_p), \dots, \sum_{p=1}^N \phi_{pK}(x_p)\right)^\intercal,
	\end{equation}
	where $\phi_{pq}: \mathbb{R} \rightarrow \mathbb{R}$ are univariate functions.
\end{definition}

This transformation can be interpreted as placing activation functions $\phi_{pq}$ on the edges connecting the input nodes $p \in [N]$ to the output nodes $q \in [K]$ of a single layer neural network and applying summation on the output nodes.

\begin{definition}[KAN,~\citep{liu2025kan}]
    Define KAN as a neural network architecture consisting of concatenated KAN layers, where the output of the previous layer serves as the input to the next one. Let $L$ be the number of KAN layers and let an integer array $[N^{(0)}, N^{(1)}, \dots, N^{(L)}]$ be given, where $N^{(l)}$ denotes the number of nodes in the $l$-th KAN layer. The KAN output, denoted by $\mathrm{KAN}(\vec{x})$, is the composition of the individual layers:
	\begin{equation}
		\mathrm{KAN}(\vec{x}) = \Phi^{(L)} \circ \dots \circ \Phi^{(1)}(\vec{x}),
	\end{equation}
	where each $\Phi^{(l)}: \mathbb{R}^{N^{(l-1)}} \rightarrow \mathbb{R}^{N^{(l)}}$ is specified by an array of univariate functions $\{\phi_{pq}^{(l)}\}$.
\end{definition}

As per the definition above, in KAN, the univariate functions $\phi_{pq}$ are parametrized as linear combinations of basis functions. The choice of a basis can be tailored to the specific application. For example, the original KAN implementation used $B$-splines defined on a grid. Subsequent works considered wavelets, Chebyshev polynomials, and Fourier expansion.

It is crucial to emphasize that KAN does not inherit the universal representation power of KART because the inner and outer functions of KART may not be learnable in practice~\citep{poggio1989kanirrelevant, akashi2001application, poggio2020theoretical}. Therefore, there is no guarantee that a deep KAN can represent any given multivariate function. Nevertheless, KAN appears to be successful in certain applications, particularly in science. For example, KAN has outperformed MLPs in learning symbolic functions commonly found in physics~\citep{liu2025kan, liu2024kan20}. Additionally, KAN offers a significant interpretability advantage: individual activation functions can be inspected, and the network can be pruned by removing functions that closely resemble zero functions, potentially discovering sparse compositional structures. Compared to MLPs, KANs have displayed the property of efficiency in certain cases and exhibit a lower spectral bias toward lower frequencies~\citep{wang2024expressiveness}.

\subsection{Quantum Kolmogorov-Arnold Network (QKAN)}\label{sec:QKAN}
In this section, we establish the main contribution of our work. Namely, we develop a quantum implementation of the Kolmogorov-Arnold Network (QKAN). QKAN is designed to realize the classical KAN on a quantum computer and, as such, consists solely of unitary transformations. This implementation leads to several technical differences between KAN and QKAN. Firstly, QKAN operates on block-encodings of vectors rather than directly on vectors themselves. In this representation, the vector is encoded in the diagonal of the top-left block of a unitary matrix. Secondly, due to the constraints of unitarity, the elements of the vectors are bounded in magnitude by one. Throughout this manuscript, we frequently utilize block-encodings of diagonal matrices. For clarity, we adopt the term ``diagonal block-encoding of a vector'' to refer to a block-encoding of a diagonal matrix where the diagonal elements correspond to the elements of the vector, as shown in \cref{fig:BE}.
\begin{figure}[H]
    \centering
    \includegraphics[width=0.4\textwidth]{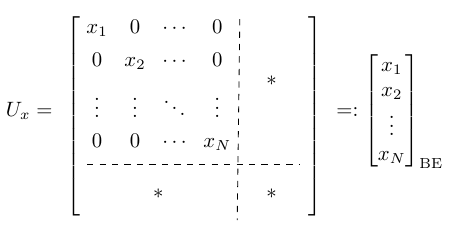}
    \caption{\textit{Diagonal block-encoding.} The top left block is a diagonal matrix whose entries are the components of an $N$-dimensional vector $\vec{x}$, while the remaining matrix blocks, denoted by asterisks ($*$), are unspecified.}
    \label{fig:BE}
\end{figure}

In this work, we limit the discussion to real values; in particular, the input, output, and weights of the QKAN model are assumed to be real. A generalization to complex numbers is possible without a significant increase in complexity by treating real and imaginary parts separately and is a direction for future work. 

We define the QKAN layer and the full QKAN model in analogy to their classical counterparts.

\begin{definition}[QKAN Layer]
    Define a QKAN layer as a transformation that, given query access to a diagonal $(1,a_x,\varepsilon_x)$-block-encoding of $\vec{x} \in [-1,1]^N$, constructs a diagonal $(1,a_x',\varepsilon_x')$-block-encoding of $\Phi(\vec{x}) \in [-1,1]^K$ such that
	\begin{equation}
		\Phi(\vec{x}) = \left( \frac{1}{N}\sum_{p=1}^N \phi_{p1}(x_p), \dots, \frac{1}{N}\sum_{p=1}^N \phi_{pK}(x_p) \right)^\intercal,
	\end{equation}
	where $\phi_{pq}: [-1,1] \rightarrow [-1,1]$ are univariate functions.
\end{definition}

\begin{definition}[QKAN]
    Define QKAN as a composition of QKAN layers, where the block-encoding produced by one layer serves as the input to the next one. Let $L$ be the number of layers in the QKAN architecture and let an integer array $[N^{(0)}, N^{(1)}, \dots, N^{(L)}]$ be given, where $N^{(l)}$ represents the number of nodes in the $l$-th layer. The QKAN output, denoted by $\mathrm{QKAN}(\vec{x})$, is a diagonal block-encoding of an $N^{(L)}$-dimensional vector, constructed recursively from the composition of layers:
	\begin{equation}
		\mathrm{QKAN}(\vec{x}) = \Phi^{(L)} \circ \dots \circ \Phi^{(1)}(\vec{x}),
	\end{equation}
	where each $\Phi^{(l)}: [-1,1]^{N^{(l-1)}} \rightarrow [-1,1]^{N^{(l)}}$ is specified by an array of univariate functions $\{\phi_{pq}^{(l)}\}$.
\end{definition}

\begin{remark}
    It is important to note that the term ``layer'' in QKAN may be slightly misleading. Unlike in classical KAN, where layers are concatenated, a QKAN layer serves as a primitive building block for the subsequent layer, leading to a recursive construction.
\end{remark}

Similarly to KAN, in QKAN the univariate functions $\phi_{pq}$ are parametrized as linear combinations of basis functions. For QKAN, Chebyshev polynomials are a natural choice of basis because they can be efficiently implemented within the qubitization framework. We define CHEB-QKAN as a QKAN where the activation functions $\phi_{pq}$ are expressed as linear combinations of Chebyshev polynomials, an idea previously explored in Ref.~\citep{ss2024chebyshevpolynomialbasedkolmogorovarnoldnetworks}:
\begin{equation}
    \phi_{pq}(x) = \frac{1}{d+1}\sum_{r=0}^d w_{pq}^{(r)} T_r(x),
\end{equation}
where $w_{pq}^{(r)} \in [-1,1]$ are linear coefficients. With these definitions, we can now state the main result of this work.

\begin{theorem}[CHEB-QKAN]
\label{theorem:cheb_qkan}
    Given access to a controlled diagonal $(1, a_x, \varepsilon_x)$-block-encoding $U_x$ of an input vector $\vec{x} \in [-1,1]^{N}$, and access to $d+1$ controlled diagonal $(1, a_w, \varepsilon_w)$-block-encodings $U_{w^{(r)}}$ of weight vectors $\vec{w}^{(r)} \in [-1,1]^{NK}$, we can construct a diagonal $(1, a_x + 1 + a_w + \log_2(d+1) + n, 4d\sqrt{\varepsilon_x} + \varepsilon_w)$-block-encoding of a vector $\Phi(\vec{x}) \in [-1,1]^K$ corresponding to the \emph{CHEB-QKAN} layer
    \begin{equation}
        \Phi(\vec{x}) = \left( \frac{1}{N}\sum_{p=1}^{N}\phi_{p1}(x_p), \dots, \frac{1}{N}\sum_{p=1}^{N}\phi_{pK}(x_p)\right)^\intercal,
    \end{equation}
    where $d$ is the maximal degree of Chebyshev polynomials used in the parameterization of activation functions $\phi_{pq}$, using $\mathcal{O}\left(d^2\right)$ applications of controlled-$U_x$ and controlled-$U_{w^{(r)}}$ and their adjoint versions.
\end{theorem}

In the above theorem, we construct QKAN using Chebyshev polynomials. However, we emphasize that QKAN is not limited to this particular basis set. In fact, due to the versatility of the QSVT framework, quantum implementations of other versions of KANs can be realized using efficient polynomial approximation of a wide range of basis functions~\citep{chao2020finding, dong2021efficient, martyn2021grand, gilyen2019quantum}. For example, to implement the $B$-spline construction for KANs in \citet{liu2025kan}'s original paper, each individual spline can be implemented by first separating each piecewise section and then taking its sum by the linear combinations of unitaries method (LCU)~\citep{berry2015simulating}. Each individual piecewise polynomial can then be implemented by multiplying a polynomial function constructed via QSVT with a threshold function formed by the sum of two Heaviside functions, which can in turn be approximated using polynomial approximations to the $\erf$ function via Lemma 10 and Corollary 5 of Ref.~\citep{low2017hamiltonian}. 

\subsection{Implementing CHEB-QKAN}\label{sec:implementation}

We prove the main~\cref{theorem:cheb_qkan} by presenting a detailed construction of CHEB-QKAN in \cref{subsection:qkan_construction}. Our construction of CHEB-QKAN relies on three basic operations: addition, multiplication, and QSVT. In a nutshell, we implement parametrized activation functions between nodes of two layers by taking linear combinations of a fixed set of basis functions where each basis function is realized through QSVT. After having constructed a single CHEB-QKAN layer that transforms a diagonal block-encoding of an $N$-dimensional input vector into a diagonal block-encoding of a $K$-dimensional output vector, the obtained block-encoding can be used as the input to the next layer by serving as the starting point for the next layer’s construction. One can immediately see that recursively transforming block-encodings in this manner results in a gate complexity that grows exponentially with the number of layers. This is because every output block-encoding is used as the elementary building block in the subsequent layer. Additionally, the total number of auxiliary qubits required increases linearly with the number of layers $L$. Finally, if the block-encodings of the input and weights are non-perfect, the error propagates recursively with every new layer, resulting in an amplified error in the output block-encoding. In summary, the recursive error propagation  and the exponential dependency of circuit depth on the number of layers limit QKAN to a shallow, i.e., $L = \mathcal{O}(1)$, albeit wide, architecture. These considerations are made precise in Supplementary Note 2.

\subsection{Application I: Quantum learning model}\label{sec:application_learning_model}
In this section, we introduce QKAN as an end-to-end quantum learning framework and outline its input and output models. In \cref{subsec:parametrization,subsec:train,subsec:interpret} we further detail the parameterization, training, and interpretability of QKAN.

By \cref{theorem:cheb_qkan}, QKAN implements a unitary whose diagonal entries block-encode a $K$-dimensional output vector. To recover these outputs classically, we apply the unitary to a quantum computational basis state and estimate a designated amplitude that encodes a multivariate function of the input that can serve for regression or classification. An end-to-end quantum speedup arises when the quantum implementation of this multivariate function requires exponentially fewer resources than a classical algorithm. The four core components enabling this speedup in QKAN are the quantum input encoding, the parametrization via Chebyshev expansions, the training algorithm, and the output extraction. Here we focus on the single-layer CHEB-QKAN case; extending to an $L$-layer CHEB-QKAN -- and to general QKAN architectures -- proceeds similarly, with an additional exponential dependence on $L$ as discussed in Supplementary Note 2.

Constructing a diagonal block-encoding of a generic $N$-dimensional classical vector requires at least $\mathcal{O}(N)$ gates~\citep{plesch2011quantum}. Because QKAN’s complexity is measured in queries to the input block-encoding, we must therefore restrict to inputs that admit efficient block-encodings that can be prepared in $\mathcal{O}(\mathrm{polylog}(N))$ time. A natural setting is when the input is inherently quantum: for example, a unitary produced by a variational quantum algorithm that prepares an approximate ground state; the time-evolution unitary $e^{-iHt}$ encoding dynamical information; a block-encoding of a quantum Gibbs state $e^{-\beta H}$; or a block-encoding of a Hamiltonian $H$ via the LCU method. Given such a unitary, we propose two methods for efficient diagonal block-encoding of the input vector. The first method treats $U$ as an efficiently implementable state-preparation unitary and constructs a diagonal block-encoding of the amplitude vector $|\psi\rangle_n = U|0\rangle_n$ (\cref{subsec:input_block_encoding} \cref{lemma:diagonal_block_encoding}). As a result, QKAN computes a multivariate function of the quantum amplitudes. The second method forms the diagonal block-encoding by taking the Hadamard (entry-wise) product with the identity $U \circ I$, retaining only the diagonal entries of the unitary (\cref{subsec:input_block_encoding} \cref{lemma:removal_of_non_diagonal}).

To estimate the values within the diagonally block-encoded output vector $\mathrm{QKAN}(\vec x)$ to additive error $\delta$, one can leverage the Hadamard test~\citep{cleve1998quantum}, using $\mathcal{O}(1/\delta^2)$ queries to the controlled diagonal block-encoding $U_{\mathrm{QKAN}}$. Specifically, to obtain the value of the $q$-th entry of the output vector, we prepare the state $\ket{\psi} = (H\otimes I_{\mathrm{aux}+k})CU_{\mathrm{QKAN}}(H\otimes I_{\mathrm{aux}+k})\ket{0}\ket{0}_{\mathrm{aux}}\ket{q}_k$ and estimate the expectation value $\langle Z \rangle$ of the top qubit, as depicted in \cref{figHadamardTest}. As a result, we obtain $\mathrm{QKAN}(\vec{x})_q = \re(\bra{0}_{\mathrm{aux}}\bra{q}_kU_{\mathrm{QKAN}}\ket{0}_{\mathrm{aux}}\ket{q}_k)$ to $\delta$-precision. In addition, the Hadamard test can be combined with amplitude estimation~\citep{Brassard_2002} to reduce the number of queries to $CU_{\mathrm{QKAN}}$ to $\mathcal{O}(1/\delta)$. \cref{theorem:amplitude_estimation_cheb_qkan} makes this statement precise, with the proof deferred to Supplementary Note 3.
\begin{figure}[H]
\centering
    \includegraphics[width=0.3\linewidth]{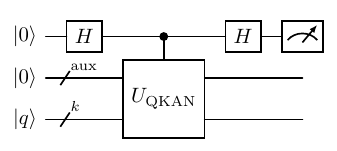}
    \caption{\textit{Circuit for solution extraction via Hadamard test.} By estimating the expectation value of Pauli-Z on the top qubit, the circuit retrieves the value $U_{\mathrm{QKAN},qq}$ to additive $\delta$-precision.}
    \label{figHadamardTest}
\end{figure} 
\begin{theorem}[Output estimation of CHEB-QKAN]
    \label{theorem:amplitude_estimation_cheb_qkan}
    Given access to a controlled diagonal $(1, a_x, \varepsilon_x)$-block-encoding $U_x$ of an input vector $\vec{x} \in [-1,1]^{N}$, and access to $d+1$ controlled diagonal $(1, a_w, \varepsilon_w)$-block-encodings $U_{w^{(r)}}$ of weight vectors $\vec{w}^{(r)} \in [-1,1]^{NK}$, we can estimate the value $\Phi(\vec{x})_q = \frac{1}{N}\sum_{p=1}^{N} \phi_{pq}(x_p)$ of the $q$-th component of the $\emph{CHEB-QKAN}$ layer to $\left( 4d\sqrt{\varepsilon_x} + \varepsilon_w + \delta\right)$-precision using $\mathcal{O}\left(d^{2}/\delta\right)$ applications of controlled-$U_x$ and controlled-$U_{w^{(r)}}$ and their adjoint versions. 
\end{theorem}

Any potential quantum speed–up is contingent on the cost of reading out the $K$ entries produced by the final layer. First, the classical post-processing cost must remain sub-exponential. We therefore restrict the output dimension to $K = \mathcal{O}(\mathrm{polylog}(N))$, since estimating an exponential number of values would itself take exponential time. Fortunately, setting $K=\mathcal{O}(1)$ is already sufficient for most regression and classification tasks. Second, consider the precision with which each amplitude $\alpha_q = \frac{1}{N}\sum_{p=1}^{N}\phi_{pq}(x_p)$ is estimated. Using the estimation procedure described in \cref{theorem:amplitude_estimation_cheb_qkan}, an additive $\delta$ approximation requires $\mathcal{O}\left(d^{2}/\delta\right)$
queries, independent of $|\alpha_q|$. If a multiplicative (relative) error is required, i.e., $|\hat\alpha_q-\alpha_q|<\delta\,|\alpha_q|$, the query count increases to $\mathcal{O}\left(d^{2}/(\delta\,|\alpha_q|)\right)$, because the amplitude must now be resolved to a fixed fraction of its value. Consequently, the potential quantum speed-up is preserved as long as $\alpha_q$ does
not decay exponentially, that is, provided  
$|\alpha_q|^{-1} = \mathcal{O}(\mathrm{polylog}(N))$. The last requirement is not an artifact of QKAN; it is analogous to the inverse-amplitude overhead in amplitude-amplification/estimation, which scales as $1/\sqrt{a}$ with the marked-state probability $a$~\citep{Brassard_2002}, and black-box state-preparation, that scales as $1/\mathcal F$ with the $\ell_2$ filling fraction $\mathcal F$ of the target function \citep{mcardle2022quantum}.

\subsection{Application II: Multivariate state preparation}\label{sec:application_state_preparation}

On the other hand, given the quantum nature of our algorithm, QKAN can also output a quantum state as the solution. While QKAN can be seen as a machine learning model, the algorithm itself leads to a form of multivariate state preparation~\citep{rosenkranz2025quantum}. The resulting block-encoding can be applied to the uniform superposition $\ket{+}_k := H_k\ket{0}_k$, in combination with amplitude amplification~\citep{yoder2014fixed, Brassard_2002}, to produce a quantum state with amplitudes encoding multivariate functions of the input. In the following, we show how the compositional framework of QKAN introduced in \cref{sec:implementation} can be used to prepare quantum states encoding multivariate functions on a $D$-dimensional regular grid. Specifically, we work through the special case of a $D$-dimensional Gaussian. Efficient Gaussian state preparation has been studied extensively, ranging from Grover–Rudolph~\citep{grover2002creating} and Kitaev–Webb~\citep{kitaev2008wavefunction} to more recent approaches improving depth and ancilla costs~\citep{rattew2021efficient, kane2024nearly, kuklinski2025simpler}. Our use of the Gaussian example is not meant to compete with these tailored methods but to illustrate how QKAN’s modular, compositional structure enables the assembling of multivariate amplitudes from elementary components. Finally, in \cref{subsec:state_preparation_via_cheb_qkan}, we remove the grid restriction and show that any CHEB-QKAN layer -- acting on an arbitrary input register -- yields a valid multivariate state-preparation routine.

Our aim is to prepare a quantum state of the form $\sum_{i_1,\dots,i_D}f\bigl(x_{(i_1,\dots,i_D)}\bigr)\,\lvert i_1,\dots,i_D\rangle$ with $x_{(i_1,\dots,i_D)}=\bigl(-1 + i_j\,s\bigr)_{j\in [D]} \;\in[-1,1]^D$ and ${(i_1,\dots,i_D)\in\{0,\dots,2^n-1\}^D}$. To do this, we first encode the vectorized $D$-dimensional grid points $x_{(i_1,\dots,i_D)}$ as a diagonal operator $G_D$, treating the vectorized grid as the classical input $\vec{x}$ to QKAN. In Supplementary Note 4.A, we provide the proof of \cref{lemma:multivariate_grid} by extending the one-dimensional construction of \citet{rosenkranz2025quantum} to $D$ dimensions.

\begin{lemma}[Multivariate grid encoding]
\label{lemma:multivariate_grid}
Let $G_D \;=\;\diag\bigl(x_{(i_1,\dots,i_D)}\bigr)_{(i_1,\dots,i_D)\in\{0,\dots,2^n-1\}^D}$, where
\begin{equation}
    x_{(i_1,\dots,i_D)}
=\bigl(-1 + i_1\,s,\;-1 + i_2\,s,\;\dots,\;-1 + i_D\,s\bigr)
\;\in[-1,1]^D,
\end{equation}
be a uniform (vectorized) $D$-dimensional grid on $[-1,1]^D$ with step size $s = \tfrac{2}{2^n-1}$ in every direction. The dimension of $G_D$ is $D\,2^{nD}$. Then,
\begin{equation}
    G_1 = \sum_{i=1}^{n}\Bigl(\frac{2^{i-1}}{2^n-1}\Bigr) \, I_{i-1}\otimes XZX \otimes I_{n-i} \quad \text{and} \quad G_D \;=\; \sum_{j=1}^{D} I_{n(j-1)} \otimes G_1 \otimes I_{n(D-j)}\otimes \ketbra{j}{j},
\end{equation}
and we can create a $(1, D \lceil \log n\rceil,0)$-block-encoding of $G_D$ using $\mathcal{O}(Dn(\log n + \log D))$ two-qubit gates.
\end{lemma}

In the following, we illustrate how the QKAN architecture can be used to provide a multivariate Gaussian quantum state. The main result is summarized in \cref{theorem:multivariate_gaussian}:

\begin{theorem}
    \label{theorem:multivariate_gaussian}
    We can prepare a $Dn$-qubit quantum state $\ket{\psi}$ with amplitudes corresponding to a $D$-dimensional Gaussian distribution on a regular square grid of size $(2^n)^D$ such that 
    \begin{equation}
        \left\|\ket{\psi} \; - \; \frac{1}{\widetilde F_{\rm exp}}\sum_{i_1,\dots,i_D}^{2^n}\exp(-\tfrac{\beta}{2} \sum_{j=1}^D x_{i_j}^2)\,\ket{i_1,\dots,i_D} \right\|_2 \le \delta,
    \end{equation}
    where $\widetilde F_{\rm exp}$ normalizes the target state. The procedure succeeds with arbitrarily high probability by using $\mathcal{\widetilde O}\bigl( \beta^{\tfrac{D}{4}+\tfrac{1}{2}} \, n \, \log \tfrac{1}{\delta} \bigr)$ two-qubit gates and $D\lceil \log n\rceil + \lceil \log D\rceil + 4$ ancilla qubits.
\end{theorem}

Here, the $\mathcal{\widetilde O}$ notation suppresses the logarithmic factors $\log n$ and $\log \beta$, and treats $D$ as a constant. The details of gate and qubit complexity can be found in Supplementary Note 4.D. The proof of \cref{theorem:multivariate_gaussian} is an explicit four-step construction, and the entire procedure can be viewed as an instance of a two-layer QKAN architecture.  For example, ~\cref{fig:2d_gaussian} gives an illustration for the $D=2$ case: starting from a vectorized 2D grid of points, the first layer computes $x_i^2+y_j^2-1$, and the second layer applies a polynomial approximation of the exponential to produce the Gaussian output state.

\begin{figure}[H]
\centering
    \includegraphics[width=\linewidth]{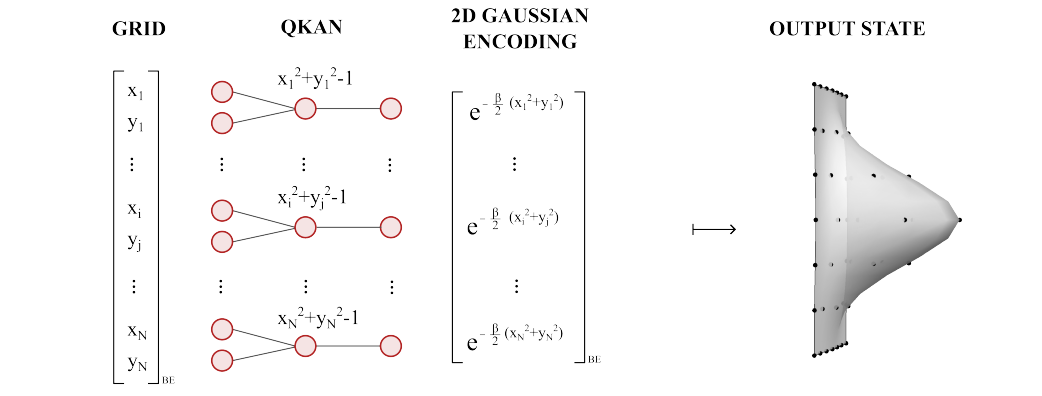}
    \caption{\textit{Example: 2D Gaussian state preparation via QKAN.} Starting from a vectorized 2D grid of points $\{(x_i,y_i)\}$ encoded as a diagonal block-encoding (left), the first QKAN layer applies Chebyshev polynomial $T_2$ and sums over the two dimensions, computing $\tfrac{1}{2}(T_2(x_i) + T_2(y_j)) = x_i^2+y_j^2-1$. A second layer uses a polynomial approximation of the exponential $e^{-\tfrac{\beta}{2}(x+1)}$ to block-encode the values $e^{-\tfrac{\beta}{2}(x_i^2+y_j^2)}$. Finally, applying this block-encoding to the uniform superposition and amplitude amplifying yields the desired 2D Gaussian distribution (right).}
    \label{fig:2d_gaussian}
\end{figure} 

The vectorized $D$-dimensional grid serves as the input to the first layer. By \cref{lemma:multivariate_grid}, we start by creating a $(1, D \lceil \log n\rceil,0)$-block-encoding $U_{G_D}$ of $G_D$ using $\mathcal{O}(Dn(\log n + \log D))$ two-qubit gates. The first layer applies the activation function $T_2(x) = 2x^2-1$ to all entries of the multivariate grid, followed by the summation over the $D$ dimensions, realizing the transformation $G_D\mapsto \Phi_1(G_D)$. Specifically, by \cref{theorem:qet_of_hermitian_matrices}, we first construct a $(1, D \lceil \log n\rceil+1,0)$-block-encoding of $T_2(G_D)$ according to the \textit{CHEB} step in the QKAN construction:
\begin{equation}
    G_D \mapsto T_2(G_D) \;=\; \sum_{j=1}^{D} I_{n(j-1)} \otimes T_2(G_1) \otimes I_{n(D-j)}\otimes \ketbra{j}{j}.
\end{equation}
This step uses $\mathcal{O}(1)$ queries to the controlled version of $U_{G_D}$ and $\mathcal{O}(D \log n)$ other two-qubit gates. We then sum over the $D$ dimensions by applying Hadamards on the last $k$ qubits of $T_2(G_D)$ and absorbing them into the auxiliary register, according to the \textit{SUM} step: 
\begin{equation}
    T_2(G_D) \mapsto \Phi_1(G_D) \; = \; \bigl(I_{Dn}\otimes \bra{0}_k H_k\bigr) T_2(G_D) \bigl(I_{Dn}\otimes H_k\ket{0}_k \bigr) \;=\; \tfrac{1}{D}\sum_{j=1}^{D} I_{n(j-1)} \otimes T_2(G_1) \otimes I_{n(D-j)}.
\end{equation}
The above transformation yields a $(1, D \lceil \log n\rceil + \lceil \log D\rceil + 1,0)$-block-encoding of $\Phi_1(G_D)$. In component-wise form, the transformation is $x_{(i_1,\dots,i_D)} \mapsto \tfrac{2}{D}\sum_{j=1}^D (-1 + i_j\,s)^2 -1$. The layer dimension is reduced from $\text{dim}\,G_D = D2^{nD}$ to $\text{dim}\,\Phi_1(G_D) = 2^{nD}$.  
The second layer implements the exponential decay:
\begin{equation}
    \Phi_1(G_D) \mapsto \Phi_2(G_D) \; \approx \; \exp \bigl(-\tilde\beta \bigl[\Phi_1(G_D) + 1\bigr] \bigr) \quad \text{with} \quad \tilde\beta \; = \; \tfrac{D}{4}\beta.
\end{equation}
In Supplementary Note 4.B we show that one can find an approximating polynomial $P_d(x)$ with degree $d = \mathcal{O}\bigl(\sqrt{D \beta}\log \tfrac{1}{\varepsilon} \bigr)$ such that $|P_d(x)-e^{-\tilde\beta(x+1)}|\le\varepsilon$ on $[-1,1]$. Such a polynomial can be realized by QSVT invoking \cref{theorem:qet_of_hermitian_matrices} using $\mathcal{O}\bigl(d)$ queries to the block-encoding of $\Phi_1(G_D)$, constructed in the previous step, and $\mathcal{O}\bigl(d \times (D \log n + \log D)\bigr)$ other two-qubit gates. By \cref{theorem:qet_of_hermitian_matrices}, we obtain a $(1, D \lceil \log n\rceil + \lceil \log D\rceil + 3,0)$-block-encoding of $\Phi_2(G_D)$. In component-wise form, the transformation is 
\begin{equation}
    \tfrac{2}{D}\sum_{j=1}^D (-1 + i_j\,s)^2 -1 \; \mapsto \; P_d\Bigl(\tfrac{2}{D}\sum_{j=1}^D (-1 + i_j\,s)^2 - 1\Bigr) \; \approx \; \exp \Bigl( -\tfrac{\beta}{2} \sum_{j=1}^D (-1 + i_j\,s)^2\Bigr).
\end{equation}
Therefore, the diagonal entries of $\Phi_2(G_D)$ correspond to the amplitudes of a $D$-dimensional Gaussian up to a maximal error $\varepsilon$. In Supplementary Note 4.D, we show that $d$ must be chosen as a function of $D$, $\beta$, and $\delta$ to obtain the target state preparation accuracy $\delta$: $d = \mathcal{O}\bigl(\sqrt{D\beta}\,\log \beta^{\tfrac{D}{4}} \log \tfrac{1}{\delta} \bigr)$. By applying the block-encoding to the uniform superposition $\ket{+}_{Dn}$ and post-selecting on the auxiliary register being in the $\ket{0}_a$ state, we prepare the desired state 
\begin{equation}
    \ket{\psi} = \frac{\Phi_2(G_D)\ket{+}_{Dn}}{\left\|\Phi_2(G_D)\ket{+}_{Dn}\right\|_2}
\end{equation}
probabilistically, with success probability $p = \left\|\Phi_2(G_D)\ket{+}_{Dn}\right\|_2^2$. We can boost $p$ to an arbitrarily high success probability by using fixed-point amplitude amplification with $\mathcal{O}(1/\sqrt{p})$ queries to the controlled block-encodings of $\Phi_2(G_D)$ (and their adjoint versions)~\citep{yoder2014fixed}. In Supplementary Note 4.C we show that $1/\sqrt{p}$ can be upper-bounded by $\mathcal{\widetilde O}(\beta^{\tfrac{D}{4}})$ using a lower bound on the continuous version of the $\ell_2$-filling fraction of the $D$-dimensional Gaussian. The total gate complexity arises from $\mathcal{\widetilde O}(\beta^{\tfrac{D}{4}})$ queries to the (controlled) block-encodings of $\Phi_2(G_D)$, as shown in Supplementary Note 4.D.

This compositional approach readily extends beyond Gaussian amplitudes to a broader class of multivariate functions. Any target map
\begin{equation}
    f(x_{(i_1,\dots,i_D)}) \;=\; g_L\bigl(g_{L-1}(\cdots g_1(x_{(i_1,\dots,i_D)}))\bigr),
\end{equation}
with each $g_i$ admitting an efficient polynomial approximation of degree $d_i$, can be implemented by cascading $L$ QKAN layers. This modularity allows one to leverage known polynomial expansions for elementary functions, such as $\sin(x)$, $\exp(x)$, or $\log(x)$, and assemble them into more complex amplitudes via successive QKAN layers. Crucially, because QKAN composes recursively by invoking the block-encoding constructed in the previous layer as the elementary building block for the next one, the overall two-qubit-gate cost scales multiplicatively as $\mathcal{O}\bigl(d_1\,d_2\cdots d_L\bigr)$, making it essential to keep the compositional depth $L$ shallow. 

Finally, in addition to the fully explicit state-preparation constructions presented above, any intermediate QKAN layer---when viewed as a parametrized quantum learning model with trainable activation functions---can itself produce a parametrized quantum state when applied to a uniform superposition; this provides a variational multivariate state-preparation method for generic functions, as formalized in~\cref{subsec:state_preparation_via_cheb_qkan}.

\section{Discussion}\label{sec:discussion}
In this work, we have defined and implemented a quantum version of the recently proposed KAN architecture in Ref.~\citep{liu2025kan}. Our proposed QKAN architecture is built on block-encodings, where both the input and the output are block-encodings. More specifically, it employs a recursive construction where block-encodings obtained in the previous layer serve as a primitive building block in the next layer. The potential applications of QKAN will be reliant on the availability of efficient block-encodings of the inputs. Specifically, we demonstrated that QKAN can serve as a quantum learning model by giving an explicit construction for encoding and training its parameters. Moreover, we demonstrated that QKAN has a broader algorithmic utility by serving as a multivariate state preparation protocol, exemplified by an explicit construction of a $D$-dimensional Gaussian distribution.

The QKAN architecture has several strengths. Firstly, it depends linearly on the cost of constructing block-encodings of input and weights. This dependency can lead to efficient implementation procedures assuming efficient block-encoding methods. Additionally, we propose that QKAN is potentially suitable for direct quantum input, for instance, quantum states whose analysis would be intractable classically. For example, in phase classification tasks~\citep{carrasquilla2017machine}, if a state corresponding to an unknown quantum phase of a physical system can be prepared efficiently, one may attempt to train QKAN in a supervised manner to classify the phase of that state. This classification would be achieved by computing a multivariate function of the state's amplitudes, potentially leading to the discovery of new order parameters. In addition, QKAN is a versatile architecture that admits different ways to encode data, parameterize weights, and perform training. Lastly, its underlying mechanism, the QSVT framework, allows the implementation of different sets of basis functions tailored to different applications.

On the other hand, QKAN exhibits several caveats. Firstly, it inherits limitations from the classical KANs, whose full potential remains to be established, even though symbolic regression tasks in science applications are brought forth. Secondly, the query complexity of a multilayer QKAN scales exponentially in the number of layers, limiting QKAN to a shallow, albeit wide architecture. Finally, some caveats arise due to the nature of the quantum subroutines involved. Quantum computers are good at representing polynomials, but not all functions can be efficiently approximated by polynomials. Therefore, the available basis functions for QKAN may generally be less powerful in approximating arbitrary functions directly compared to, for example, using splines as basis functions.

Within the broader quantum machine learning literature and even quantum algorithms, QKAN is a novel result in multiple directions. Firstly, it departs from variational architectures and steps into the more powerful quantum linear algebra toolset. From another perspective, it brings the aspect of parameterization into fault-tolerant subroutines for quantum machine learning. Secondly, it is a quantum learning model built on a new paradigm --- that of having a decomposition of a function into single transformed features and summations. Finally, QKAN enables a version of multivariate state preparation, thereby serving as an algorithmic subroutine. We hope that this work will serve as a motivation to further investigate our QKAN architecture and other types of QKAN architecture, and build quantum models beyond near-term techniques.

\section{Methods}\label{sec:methods}

\subsection{Block-encoding and quantum subroutines}
\label{secQuantumPrelim}

In the following, we outline some known results that are used in different parts of the construction and parameterization of QKAN. We begin by formally defining block-encoding:
\begin{definition}[Block-encoding -- Definition 24,~\citep{gilyen2019quantum}, see also~\citep{low2019hamiltonian, chakraborty2019power}]\label{Definition:Block_encoding}
    Let $A$ be an $n$-qubit matrix, $\alpha, \varepsilon \in \mathbb{R}_+$ and $a \in \mathbb{N}$. We say that the $(n+a)$-qubit unitary $U$ is an $(\alpha,a,\varepsilon)$-block-encoding of $A$ if
	\begin{equation}
		\|A-\alpha(\bra{0}_a\otimes I_n)U(\ket{0}_a\otimes I_n)\| \leq \varepsilon.
	\end{equation}
\end{definition}

Given block-encodings of operators $A_i$, we can construct a block-encoding of their linear combination using an auxiliary tool known as a ``state-preparation pair''. Recall that $\|\cdot\|_1$ is the $\ell_1$/Manhattan norm.

\begin{definition}[State preparation pair -- Definition 28,~\citep{gilyen2019quantum}]
    \label{Definition:State_preparation_pair}
    Let $\vec{y} \in \mathbb{C}^m$ and $\|\vec{y}\|_1 \leq \beta$. The pair of unitaries $(P_L,P_R)$ is called a ($\beta,b,\varepsilon_{\mathrm{SP}}$)-state-preparation-pair for $\vec{y}$ if
    \begin{align}
        P_L\ket{0^b} = \sum_{j=1}^{2^b} c_j\ket{j}, \quad
        P_R\ket{0^b} = \sum_{j=1}^{2^b} d_j\ket{j},
    \end{align}
    such that $\sum_{j=1}^{m} |y_j-\beta c_j^*d_j| \leq \varepsilon_{\mathrm{SP}}$ and $c_j^*d_j = 0$ for $j=m+1,\dots,2^b$.
\end{definition}
One can think of a state preparation pair as encoding the desired state/vector $\vec{y}$ in the first $m$ elements of a length-$2^b$ column vector whose elements are $c_j^*d_j$, up to an error of $\varepsilon_{\mathrm{SP}}$. The role of $\beta$ is to take care of normalization.

\begin{lemma}[Linear combination of block-encodings -- Lemma 29,~\citep{gilyen2019quantum}]
    \label{lemma:GSLW_linear_combination_BE}
    Let $A = \sum_{j=1}^m y_jA_j$ be an $n$-qubit operator and $\varepsilon \in \mathbb{R}^+$. Suppose that $(P_L, P_R)$ is a $(\beta, b, \varepsilon_1)$-state-preparation-pair for $\vec{y}$ and 
    \begin{equation}
        W = \sum_{j=1}^{m} |j\rangle\langle j| \otimes U_j + ((I - \sum_{j=1}^{m} |j\rangle\langle j|) \otimes I_a \otimes I_n)
    \end{equation}
    is an $n + a + b$ qubit unitary such that for all $j \in {1, \dots, m}$ we have that $U_j$ is an $(\alpha, a, \varepsilon_2)$-block-encoding of $A_j$. Then we can implement a $(\alpha\beta, a + b, \alpha\varepsilon_1 + \beta\varepsilon_2)$-block-encoding of $A$, with a single use of $W$, $P_R$ and $P_L^\dagger$.
\end{lemma}

We can also construct a block-encoding of a product of two block-encoded matrices.
\begin{lemma}[Product of block-encodings -- Lemma 30,~\citep{gilyen2019quantum}]
    \label{lemma:product_BE}
    Let $U_A$ be a $(\alpha,a,\varepsilon_A)$-block-encoding of $A$ and $U_B$ be a $(\beta,b,\varepsilon_B)$-block-encoding of $B$, where $A,B$ are $n$-qubit operators. Then, $(I_b \otimes U_A)(I_a \otimes U_B)$ is a $(\alpha\beta,a+b,\alpha\varepsilon_B+\beta\varepsilon_A)$-block-encoding of $AB$.
\end{lemma}

In~\cref{lemma:product_BE}, identity operators act on each other’s auxiliary qubits, slightly abusing the notation. Formally, $I_b \otimes U_A := I_b \otimes \sum_{i,j=1}^{2^a} \ket{i}\bra{j} \otimes (U_A)_{ij}$, while $I_a \otimes U_B := \sum_{i,j=1}^{2^b} \ket{i}\bra{j} \otimes I_a \otimes (U_B)_{ij}$.

\begin{lemma}[Hadamard product of block-encodings -- Theorem 4,~\citep{guo_quantum_2024}]
    Let $U_A$ be a $(\alpha,a,\varepsilon_A)$-block-encoding of $A$ and $U_B$ be a $(\beta,b,\varepsilon_B)$-block-encoding of $B$, where $A,B$ are $n$-qubit operators. Then
    \begin{equation}
        \left( P \otimes I_{a+b} \right) U_A \otimes U_B \left( P^\dagger \otimes I_{a+b} \right)
    \end{equation}
    is a $(\alpha\beta,a+b+n,\alpha \varepsilon_B + \beta \varepsilon_A)$-block-encoding of $A \circ B$. Here $U_A$ and $U_B$ are the block-encodings of $A$ and $B$ respectively and $P := \sum_{i,j=1}^{N} \ket{i}\bra{i} \otimes \ket{i \oplus j}\bra{j}$ can be constructed using $n$ CNOT gates, namely one CNOT gate between each pair of corresponding qubits from the first and second registers.
    \label{lemma:hadamard_product}
\end{lemma}

Polynomial functions can be applied to the singular values of a block-encoded matrix (or eigenvalues of a Hermitian matrix) through quantum signal processing (QSP)~\citep{low2017optimal, low2019hamiltonian} and QSVT~\citep{gilyen2019quantum}. For a more detailed account of QSVT and its applications in quantum algorithms, we refer the reader to Refs.~\citep{gilyen2019quantum, martyn2021grand,dalzell2023quantum}.

\begin{lemma}[Constructing Chebyshev polynomials via QSP -- Lemma 6,~\citep{gilyen2019quantum}] 
    \label{lemma:chebyshev_qsp}
Let $T_d\in\mathbb{R}[x]$ be the $d$-th Chebyshev polynomial of the first kind. Let $\Phi\in\mathbb{R}^d$ be such that $\phi_1=(1-d)\frac\pi2$, and for all $i\in[d]\setminus\{1\}$, let $\phi_{i}:=\frac\pi2$. Then 
	\begin{equation}
 \prod_{j=1}^{d}\left(e^{i\phi_j\sigma_z}R(x)\right)
	=\left[\begin{array}{cc} T_d(x) & . \\ . & . \end{array}\right],\text{ where }R(x) := \begin{bmatrix} x & \sqrt{1-x^2} \\ \sqrt{1-x^2} & -x \end{bmatrix},
	\end{equation}
 is a $(1,1,0)$-block-encoding of $T_d(x)$.
\end{lemma}

\begin{theorem}[Polynomial Eigenvalue Transformation -- Theorem 31,~\citep{gilyen2019quantum}]\label{theorem:qet_of_hermitian_matrices}
Let $U$ be an $(\alpha,a,\varepsilon)$-encoding of a Hermitian matrix $A$ and $P \in \mathbb{R}[x]$ be a degree-$d$ polynomial satisfying $|P(x)| \leq \frac{1}{2}$ on $[-1,1]$.
Then, one can construct a quantum circuit $\tilde{U}$ which is a $(1,a+2,4d\sqrt{\varepsilon/\alpha})$-encoding of $P(A/\alpha)$. $\tilde{U}$ consists of $d$ $U$ and $U^\dag$ gates, one controlled-$U$, and $\mathcal{O}((a+1)d)$ other one- and two-qubit gates.
\end{theorem}
Note that for Chebyshev polynomials, the $a+1$ auxiliary qubits are sufficient, and the $|P(x)| \leq \frac{1}{2}$ constraint is relaxed since $|T_d(x)| \le 1$.

\subsection{CHEB-QKAN construction}\label{subsection:qkan_construction}

In the following, we present a detailed construction of CHEB-QKAN. Our model takes as input a real $N$-dimensional vector $\vec{x} = (x_1, x_2, \dots, x_{N}) \in [-1,1]^N$. We encode and process this input data on the diagonal of a matrix, assuming it is provided as a diagonal $(1, a_x, \varepsilon_x)$-block-encoding $U_x$, such that $\|\langle 0|_a U_x |0\rangle_a - \diag(x_1, \dots, x_{N})\| \le \varepsilon_x$ \footnote{Note that if $U$ is an $(\alpha,a,\varepsilon)$-block-encoding of $A$, then equivalently it is a $(1,a,\frac{\varepsilon}{\alpha})$-block-encoding of $\frac{A}{\alpha}$. Therefore, we can factor $\alpha$ into the input vector $\vec{x}$ and restrict the discussion to $\alpha = 1$.}. For example, the input could be the amplitudes of a quantum state: given access to an amplitude-encoding unitary $U$ where $U|0\rangle_n = \sum_{i=1}^{N} x_i |i\rangle_n$ with $x_i \in [-1,1]$, we can efficiently construct a diagonal $(1, n+2, 0)$-block-encoding of $\{x_i\}$, as detailed in~\cref{lemma:diagonal_block_encoding}. Alternatively, techniques for constructing exact block-encodings for sparse matrices may also be employed~\citep{camps2024explicit}. Since such constructions may not be diagonal, one first removes the off-diagonal entries of the constructed block-encoding via \cref{lemma:removal_of_non_diagonal}. In the following, we detail the steps to build a CHEB-QKAN layer.

Suppose that the layer has $K$ output nodes. Then, we need $NK$ different parametrized activation functions, one between every two input and output nodes. To accommodate this, we first dilate the input block-encoding by appending $k = \log_2 K$ auxiliary qubits (Supplementary Note 1). This produces a $(1, a_x, \varepsilon_x)$-block-encoding of 
\begin{equation}
    \diag(\underbrace{x_1, \dots, x_1}_{K}, \dots, \underbrace{x_{N}, \dots, x_{N}}_{K}) = \sum_{p=1}^N\sum_{q=1}^K x_p\ketbra{p}{p}_n\otimes\ketbra{q}{q}_k.
\end{equation}
The obtained block-encoding $U_x\otimes I_k$ serves as the foundation for subsequent operations.
\begin{figure}[H]
    \centering
    \includegraphics{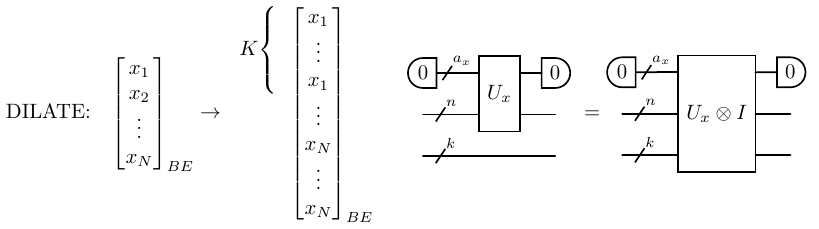}
    \caption{\textit{Step 1.} Expand the input block-encoding of the $N$-dimensional input vector $\vec{x}$ by appending $k = \log_2 K$ auxiliary qubits, resulting in a block-encoding containing $K$ copies of each vector component.}
    \label{fig:DILATE}
\end{figure}

Our trainable activation functions are linear combinations of Chebyshev polynomials of the first kind because these polynomials can be natively realized using QSVT~\citep{low2019hamiltonian}. To implement these polynomials, we alternately apply $U_{x}$ and $U_{x}^{\dagger}$, interleaved by reflection operators according to \cref{lemma:chebyshev_qsp}. By utilizing \cref{theorem:qet_of_hermitian_matrices}, we obtain a $(1, a_x+1, 4r\sqrt{\varepsilon_x})$-block-encoding of the diagonal matrix
\begin{equation}
    \diag(\underbrace{T_r(x_1), \dots, T_r(x_1)}_{K}, \dots, \underbrace{T_r(x_N), \dots, T_r(x_N)}_{K}) = \sum_{p=1}^N\sum_{q=1}^K T_r(x_p)\ketbra{p}{p}_n\otimes\ketbra{q}{q}_k,
\end{equation}
using $\mathcal{O}(r)$ applications of $U_{x}$ and $U_{x}^\dagger$ and a single auxiliary qubit. We denote this resulting block-encoding by $U_{T_r}$, which serves as the building block for our activation functions.
\begin{figure}[H]
    \centering
    \includegraphics[width=\textwidth]{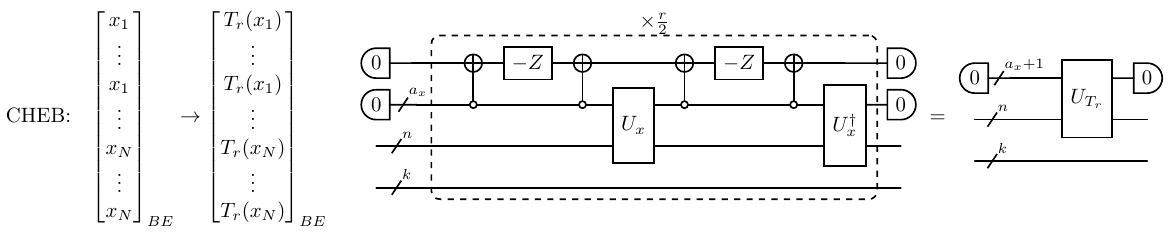}
        \caption{\textit{Step 2.} Apply Chebyshev polynomials of degrees $r \in \{1,\dots,d\}$ to the diagonal block-encoding from Step 1 by interleaving the input block-encoding $U_x$ and its adjoint $U_x^\dagger$ with reflection operators $Z_{\pi} := I\otimes(2\ketbra{0}{0}_{a_x}-I_{a_x})$. For even $r$, apply $(U_x^{\dagger} Z_{\pi} U_x Z_{\pi})^{\frac{r}{2}}$; for odd $r$, apply $U_x Z_{\pi}(U_x^{\dagger} Z_{\pi} U_x Z_{\pi})^{\lfloor\frac{r}{2}\rfloor}$~\citep{lin2022lecturenotesquantumalgorithms}. The $k$ auxiliary qubits from Step 1 remain unused in this construction, serving only to maintain the expanded dimension of the block-encoding.}
    \label{fig:CHEB}
\end{figure}

We repeat the previous to prepare separate block-encodings for each of the $d+1$ Chebyshev polynomials, with degrees $r$ ranging from 0 to $d$. Note that $T_0(x) = 1$ is trivial and corresponds to identity. Therefore, a general linear combination requires a total of $(d+1)NK$ linear coefficients. For this construction, we assume the weights are provided in the form of $d+1$ diagonal $(1, a_w, \varepsilon_w)$-block-encodings of $NK$-dimensional real weight vectors $\vec{w_r}$, denoted by $U_{w_r}$. We intentionally defer the discussion of the precise construction of these block-encodings to parametrization \cref{subsec:parametrization}, in order to maintain the generality of our approach. Applying \cref{lemma:product_BE}, we multiply each Chebyshev polynomial encoding by its respective weight block-encoding. This yields a $(1, a_x + 1 + a_w, 4r\sqrt{\varepsilon_x} + \varepsilon_w)$-block-encoding for each diagonal matrix:
\begin{equation}
    \diag(\underbrace{w_{11}^{(r)} T_r(x_1), \dots, w_{1K}^{(r)} T_r(x_1)}_{K}, \dots, \underbrace{w_{N1}^{(r)} T_r(x_N), \dots, w_{NK}^{(r)} T_r(x_N)}_{K}) = \sum_{p=1}^N\sum_{q=1}^K w_{pq}^{(r)} T_r(x_p)\ketbra{p}{p}_n\otimes\ketbra{q}{q}_k.
\end{equation}
Here, $w_{pq}^{(r)}$ is the weight of the $r$-th Chebyshev polynomial in the activation function between the $p$-th input node and the $q$-th output node. We must have $|w_{pq}^{(r)}| \le 1$ for all $p,q,r$ due to unitarity.

\begin{figure}[H]
    \centering
    \includegraphics[width=\textwidth]{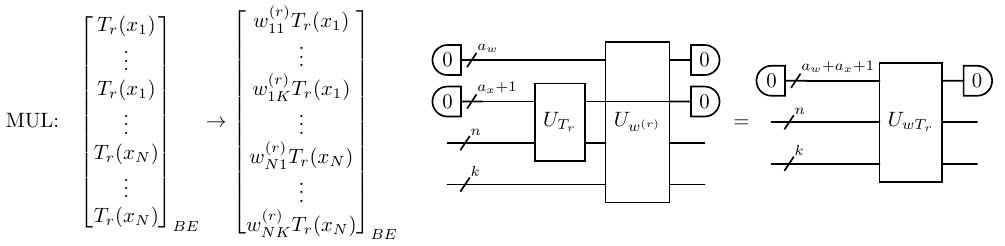}
    \caption{\textit{Step 3.} Multiply the block-encoded Chebyshev polynomials from Step 2 by an $NK$-dimensional weight vector $\vec{w^{(r)}}$. The weight $w_{pq}^{(r)}$ corresponds to the coefficient in front of the $r$-th Chebyshev polynomial in the activation function $\phi_{pq}$. In the circuit, the respective block-encodings do not "overlap" on their auxiliary qubits and, therefore, the $a_x +1$ wire goes ''through" $U_{w^{(r)}}$.}
    \label{fig:MUL}
\end{figure}

Finally, we combine the $d+1$ weighted block-encodings of Chebyshev polynomials by taking the linear combination of block-encodings with an equal superposition using \cref{lemma:GSLW_linear_combination_BE}. This process yields the desired $(1, a_x + 1 + a_w + \log_2(d+1), 4d\sqrt{\varepsilon_x} + \varepsilon_w)$-block-encoding of the diagonal matrix containing $NK$ activation functions:
\begin{equation}
    \diag(\underbrace{\phi_{11}(x_1), \dots, \phi_{1K}(x_1)}_{K}, \dots, \underbrace{\phi_{N1}(x_N), \dots, \phi_{NK}(x_N)}_{K}) = \sum_{p=1}^N\sum_{q=1}^K \phi_{pq}(x_p)\ketbra{p}{p}_n\otimes\ketbra{q}{q}_k.
\end{equation}
Here, $\phi_{pq}(x)$ denotes the activation function between the $p$-th input node and the $q$-th output node:
\begin{equation}
    \phi_{pq}(x) := \frac{1}{d+1} \sum_{r=0}^d w_{pq}^{(r)} T_r(x_p)
\end{equation}
The LCU procedure requires $\log_2(d+1)$ auxiliary qubits, assuming $d+1$ is a power of two, with the equal superposition created by applying $H_{\log_2(d+1)}$. In \cref{subsec:parametrization}, we consider generalizing this step by replacing $H_{\log_2(d+1)}$ with a parametrized unitary to control the contribution of each basis function. We denote the obtained block-encoding by $U_{\phi}$.
\begin{figure}[H]
    \centering
    \includegraphics[width=\textwidth]{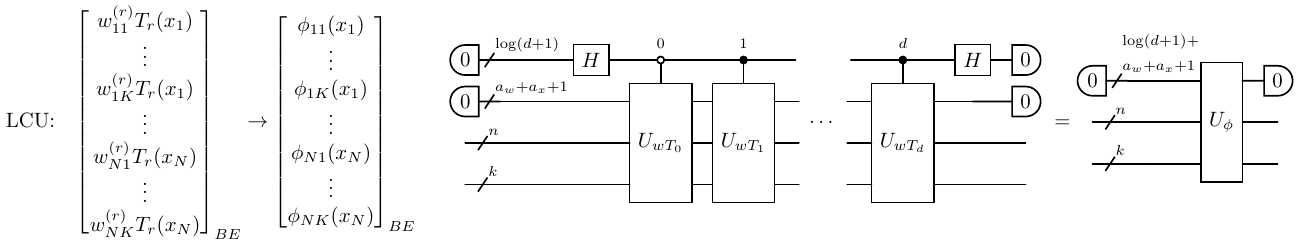}
    \caption{\textit{Step 4.} Add the $(d+1)$ block-encodings of weighted Chebyshev polynomials obtained in Step 3 using LCU. This requires $\log_2(d+1)$ control qubits, which are initialized in an equal superposition state using multi-qubit Hadamard gates.}
    \label{fig:LCU}
\end{figure}

In the final step of the construction, we want to produce a diagonal block-encoding corresponding to the output of the CHEB-QKAN layer. Starting from the block-encoding $U_{\phi}$ of individual activation functions obtained in the previous step, we apply a layer of Hadamards $H_n$ on the $n$ `input' qubits, corresponding to the summation of $N$ inputs for each of the $K$ output nodes:
\begin{equation}
    \left(H_n\otimes I_k\right) \left( \sum_{p=1}^N\sum_{q=1}^K \phi_{pq}(x_p)\ketbra{p}{p}_n\otimes\ketbra{q}{q}_k \right)\left(H_n\otimes I_k\right).
\end{equation}
This results in a block-encoding of a matrix whose diagonal elements hold the desired summation:
\begin{equation}
    \bra{0}_n\otimes I_k\left( \sum_{p=1}^N\sum_{q=1}^K \phi_{pq}(x_p)H_n\ketbra{p}{p}_n H_n\otimes \ketbra{q}{q}_k \right) \ket{0}_n\otimes I_k = \sum_{q=1}^K \left( \frac{1}{N}\sum_{p=1}^N \phi_{pq}(x_p) \right) \ketbra{q}{q}_k.
\end{equation}
In particular, the obtained unitary, denoted by $U_{\Phi}$, is a $(1, a_x + 1 + a_w + \log_2(d+1) + n, 4d\sqrt{\varepsilon_x} + \varepsilon_w)$ block-encoding of the diagonal matrix 
\begin{equation}
    \diag\left( \frac{1}{N}\sum_{p=1}^{N}\phi_{p1}(x_p), \dots, \frac{1}{N}\sum_{p=1}^{N}\phi_{pK}(x_p)\right),
\end{equation}
since
\begin{align}
    &\| \bra{0}_n\left(H_n\otimes I_k\right) \diag(\phi_{11},\dots,\phi_{NK})  \left(H_n\otimes I_k\right)\ket{0}_n - \bra{0}_n\left(H_n\otimes I_k\right) \bra{0}_{\mathrm{aux}} U_{\phi} \ket{0}_{\mathrm{aux}} \left(H_n\otimes I_k\right)\ket{0}_n\| \nonumber\\
    =& \left\| \diag\left( \frac{1}{N}\sum_{p=1}^{N}\phi_{p1}(x_p), \dots, \frac{1}{N}\sum_{p=1}^{N}\phi_{pK}(x_p)\right) - \frac{1}{N}\sum_{p=1}^{N} \left(\bra{p}_n\otimes I_k\right) \bra{0}_{\mathrm{aux}} U_{\phi} \ket{0}_{\mathrm{aux}}\left(\ket{p}_n\otimes I_k\right)\right\|\nonumber\\
    \le &\max_{p,q}|\phi_{pq}(x_p)- \bra{p}_n\bra{q}_k\bra{0}_{\mathrm{aux}} U_{\phi} \ket{0}_{\mathrm{aux}}\ket{p}_n\ket{q}_k| \nonumber\\
    \le &4d\sqrt{\varepsilon_x} + \varepsilon_w.
\end{align}

Notice that we have absorbed the $n$ ``input'' qubits into the auxiliary register.
\begin{figure}[H]
    \centering
    \includegraphics[width=\textwidth]{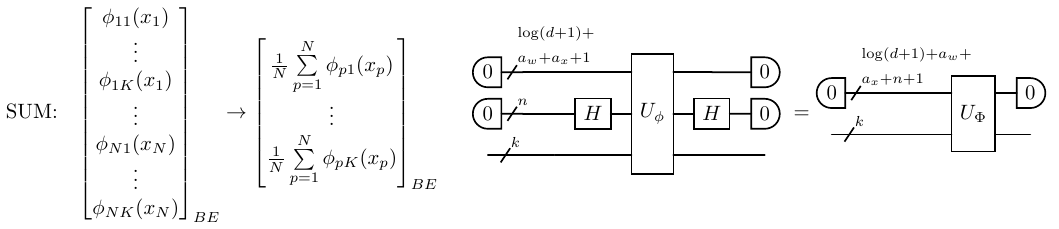}
    \caption{\textit{Step 5.} Sum the individual activation functions over $N$ input nodes for each output node, creating the desired diagonal block-encoding $U_{\Phi}$ of the $K$-dimensional output vector $\Phi(\vec{x})$. This is achieved by sandwiching the block-encoding from Step 4 with two $n$-qubit Hadamard gates. The dimension reduction occurs as the $n$ qubits originally used for input block-encoding are moved to the auxiliary register.}
    \label{fig:SUM}
\end{figure}

\subsection{Input block-encoding for the learning model}\label{subsec:input_block_encoding}

We formalize the two proposed methods for efficient diagonal block-encoding of the input vector. 

\begin{lemma}[Diagonal block-encoding of amplitudes -- Theorem 2,~\citep{rattew2023nonlineartransformationsquantumamplitudes}]
    Given an $n$-qubit quantum state specified by a state preparation unitary $U$, such that $|\psi\rangle_n = U|0\rangle_n = \sum_{j=1}^{N} \psi_j|j\rangle_n$ (with $\psi_j \in \mathbb{C}$), we can prepare a $(1, n+2, 0)$-block-encoding $U_A$ of the diagonal matrix $A = \diag(\re(\psi_1), \dots, \re(\psi_{N}))$ with $O(n)$ circuit depth and a total of $O(1)$ queries to a controlled-$U$ gate.
    \label{lemma:diagonal_block_encoding}
\end{lemma}

\begin{lemma}[Removing off-diagonal matrix elements]
    Let $U_A$ be an $(\alpha,a,\varepsilon)$-block-encoding for an $n$-qubit operator $A$. Then 
    \begin{equation}
        \left(I_{a} \otimes P\right) U_A \otimes I_n \left(I_{a} \otimes P\right)
    \end{equation}
    is a $(\alpha,a+n,\varepsilon)$-block-encoding of $\diag(A_{11},\dots,A_{NN})$, where $P = \sum_{i,j=1}^{N} \ket{i}\bra{i} \otimes \ket{i \oplus j}\bra{j}$.
    \label{lemma:removal_of_non_diagonal}
\end{lemma}

\begin{proof}
    Noting that $\diag(A_{11},\dots,A_{NN}) = A \circ I_n$, this follows immediately from~\cref{lemma:hadamard_product}, where $A \leftarrow A$ and $B \leftarrow I_n$. Note that $I_a$ acts on the $a$ auxiliary qubits of $U_A$ and that, trivially, $I_n$ is a $(1,0,0)$-block-encoding of itself. 
\end{proof}

\subsection{Parametrization of the QKAN learning model}\label{subsec:parametrization}
The parameterization of QKAN depends on being able to encode the relevant parameters into a diagonal block-encoding and update it. Rather than classically computing and reconstructing the block-encoding at each iteration -- which can be costly or technically challenging~\citep{camps2024explicit} -- we combine parametrized quantum circuits (PQCs)~\citep{benedetti2019parameterized} with our block-encoding propositions in \cref{subsec:input_block_encoding} to enable efficient updates in the \textit{MUL} step of \cref{subsection:qkan_construction}. We first provide two methods to parameterize the diagonal block-encodings of the weight vectors $\vec{w}$ used in the \textit{MUL} step in \cref{subsection:qkan_construction}, which correspond to different norm constraints, namely $\| \vec{w} \|_1 \le 1$, $\| \vec{w} \|_2 \le 1$, and $\| \vec{w} \|_\infty \le 1$.

The first method encodes the real parts of the amplitudes of a parametrized state
\begin{equation}
	\ket{w(\vec\theta)}_{n+k} = U(\vec\theta)\ket{0}_{n+k}
\end{equation}
as a diagonal $(1, n+k+2, 0)$-block-encoding via \cref{lemma:diagonal_block_encoding}. The step can be done efficiently using $O(1)$ queries to the controlled-$U(\vec\theta)$ gate and its adjoint version. The $\ell_2$ normalization the pure quantum state $\ket{w(\vec\theta)}_{n+k}$ implies a strict equality constraint $\| \vec{w} \|_2 = 1$. However, by encoding only the real parts by \cref{lemma:diagonal_block_encoding}, the constraint on the weights is relaxed to an inequality $\| \re(\vec{w}) \|_2 \le 1$. Such regularization is utilized in classical machine learning to prevent overfitting~\citep{krogh1991simple,ng2004feature}, but on the other hand, it may limit the expressibility of the QKAN model, which can be hypothesized from the effects of $\ell_2$ upper bounds of weights on Rademacher complexity~\citep{bartlett2003rademacher}. 

Alternatively, we can impose an $\ell_1$ regularization by encoding the real values of the state onto the diagonal as described above, and then squaring it by applying the block-encoding twice (\cref{lemma:product_BE}) to create a block-encoding with $\|\vec{w}^2\|_1\le 1$. This, however, places the weight within a probability simplex since $w_j^2 \ge 0$. To generalize the parameterization from the simplex to the $\ell_1$ ball, we can multiply the simplex-constrained block-encoding by an alternatively parametrized block-encoding with $\{1, -1\}$ on its diagonal, which can be prepared by applying the QSVT-approximated sign function to another parametrized block-encoding. We note that there is a limitation to the precision of the approximation of the sign function when the parameters are close to 0, and we leave further improvements on such parameterizations to future work.

The second method is to utilize \cref{lemma:hadamard_product} and replace \textit{MUL} with a Hadamard product between a parametrized unitary (or parametrized block-encoding) and our prepared diagonal block-encoding $U_{T_r}$ of the $r$-th Chebyshev polynomial. A parametrized unitary $U(\vec{\theta})$ is a $(1, 0, 0)$-block-encoding of itself. As the off-diagonal elements of the diagonal block-encoding $U_{T_r}$ are zero, so is the Hadamard product 
\begin{equation}
    U(\vec\theta) \circ \bra{0}_{\mathrm{aux}}U_{T_r}\ket{0}_{\mathrm{aux}},
\end{equation}
so the result remains a diagonal block-encoding. We then effectively create a parametrized diagonal. We conjecture that the diagonal block-encoding produced by this second method can be more expressive compared to the first method because, in principle, all diagonal entries of $U(\vec\theta)$ could be in the range of $[-1, 1]$ such that the $\ell_\infty$ norm of the weights is upper bounded by $1$: $\| \vec{w} \|_\infty \le 1$. This provides a much weaker regularization effect on the weights compared to the $\ell_2$ and $\ell_1$ regularizations in the first method, which would result in higher expressibility while retaining a regularization effect.

Given that, for the purpose of this paper, we assume that all the weights in our QKAN are real, we must restrict the diagonal entries of $U(\vec\theta)$ to be real. In principle, one can limit the use of PQCs to consist of only real gates, e.g., RY, CNOT, and CZ gates. However, such limitations would greatly restrict the expressive power of PQCs as they would be confined to orthogonal transformations. Alternatively, one can prepare the adjoint $U(\vec\theta)^\dagger$ of the PQC and combine the two using LCU, so that the diagonal entries of the resulting matrix, $\frac{1}{2}(U(\vec\theta)^\dagger + U(\vec\theta))$, are real. Finally, the parameterization can be further augmented by composing the constructions outlined above, for example, by taking sums, products, or applying polynomial functions to the parametrized diagonal block-encoding as discussed in \cref{secQuantumPrelim}. 

Thus far, we have limited the parameterization of QKAN to the weight vectors. However, parameterizing other steps of QKAN could be advantageous. For instance, in the \textit{LCU} step in \cref{subsection:qkan_construction}, instead of using Hadamard transforms, we can parameterize the unitaries used as state preparation pairs (\cref{Definition:State_preparation_pair}) to control the global weights of Chebyshev polynomials of different degrees. With a sufficiently expressive Ansatz, a training strategy could involve iteratively adding higher-degree Chebyshev polynomials to the sum and optimizing their global weights. By inspecting the optimized weights, one could determine the optimal number of Chebyshev polynomials, e.g., if the weight of the newly added polynomial vanishes. Furthermore, while we focus on Chebyshev polynomials, the QSVT framework allows implementing any bounded polynomial using \cref{theorem:qet_of_hermitian_matrices}, and many functions of interest can be well approximated by polynomials in the QSVT framework~\citep{gilyen2019quantum}. This, of course, includes splines, which can be approximated via a polynomial approximation of the $\erf$ function, as mentioned in \cref{sec:QKAN}. While a fixed basis can be selected, in principle, even the angles of QSVT could be made trainable parameters, which would allow the selection of basis functions to be part of learning as well.

\subsection{Training of the QKAN learning model}\label{subsec:train}
The parameterization methods in \cref{subsec:parametrization} remove the cost of constructing a new circuit that implements a diagonal block-encoding given that we only need to update the parameters in the PQC. Considering the full circuit with the Hadamard test and amplitude estimation, we can view the entire circuit as a very large PQC that has parameters that are repeated multiple times throughout the circuit. Given this construction, we note that the same parameter would be repeated throughout the circuit due to sequential repetitions of circuit blocks from QSVT and amplitude estimation. To achieve the analytical derivative from parameter-shift rules~\citep{mitarai2018quantum, schuld2019evaluating, schuld2020circuit}, we note that from the product rule, the gradient can be obtained from the sum of gradients of individual sub-terms, and thus can be found by computing the sum of the parameter-shifted circuits~\citep{mitarai2018quantum, schuld2019evaluating, schuld2020circuit} of gates that share the same classical parameter. Therefore, the evaluation of the gradient via parameter shift rule would then cost $\mathcal{O} (d)$ queries for single layer QKANs, and $\mathcal{O} (d^{2L})$ for $L$-layer QKANs. Note that the number of queries required to evaluate the gradient of a single parameter in the QKAN architecture also grows exponentially with the number of layers. 

Given that it is costly to obtain the full analytical gradient, we can make use of gradient estimation techniques to achieve a much more efficient estimation. Instead of perturbing each occurrence of the variable individually, one can obtain an estimate of the gradient by finite difference methods, which would only require a perturbation of shared variables once. By extension, one can also use simultaneous perturbation stochastic approximation (known as SPSA) \citep{spall1992multivariate} to produce gradient estimates with a cost unrelated to the number of parameters (both free and repeated) to achieve a much more efficient training strategy.

The circuit parameters can then be updated using optimizer algorithms such as gradient descent or Adam~\citep{kingma2015adam}. Further, one can also incorporate quantum natural gradient methods~\citep{stokes2020quantum, koczor2022quantum} to achieve faster convergence by again using parameter shift rules or, with a constant cost, SPSA, to compute the quantum Fisher information matrix~\citep{gacon2021simultaneous}.

\subsection{Interpretability of the QKAN learning model}\label{subsec:interpret}

Interpretability of KANs, as formalized in Ref.~\citep{liu2025kan}, refers to identifying and pruning unimportant branches of the model. In QKAN, this can be achieved in a manner consistent with our block‐encoding parametrization: the same parametrized quantum state that defines the weights also provides sample access to their relative importance.

Consider the first method in \cref{subsec:parametrization} where the weights are encoded from the real amplitudes of $\ket{\psi(\vec{\theta})} = U(\vec{\theta})\ket{0}$, i.e., $w_j = \re(\psi_j(\vec{\theta}))$. To obtain sampling access consistent with the encoded weights, we use a standard LCU combination of $U(\vec{\theta})$ and $U(\vec{\theta})^\dagger$ and a single ancilla qubit to prepare the state
\begin{equation}
    \ket{0}\ket{\mathrm{Re}(\psi)} + \ket{1}\ket{\mathrm{Im}(\psi)},
\end{equation}
where $\ket{\mathrm{Re}(\psi)} = \sum_j \mathrm{Re}(\psi_j)\ket{j}$. By measuring this state in the computational basis and postselecting on $\ket{0}$ in the first register, we sample indices with probability
\begin{equation}
    \Pr[j] = |\re(\psi_j(\vec{\theta}))|^2 = |w_j|^2.
\end{equation}
The postselection succeeds with probability $p_{\mathrm{succ}}=\sum_j w_j^2=\|\vec{w}\|_2^2$, which, if desired, can be increased via amplitude amplification, yielding a sampling distribution arbitrarily close to $\Pr[j] = w_j^2/\|\vec{w}\|_2^2$. As a result, indices with large $|w_j|$ are sampled more frequently, while small weights are further suppressed when squared. Thus, basis functions with vanishing weights will have low sampling probability and can be pruned after the training is completed. Therefore, by sampling the trained state-preparation circuits, we can obtain a compressed and pruned model of the trained QKAN that can be explicitly interpreted. If one additionally requires the signed values $w_j$ for selected indices $j$, they can be estimated using a standard Hadamard-test combined with amplitude estimation, analogous to the output-estimation step described in \cref{sec:application_learning_model}.

Finally, classical KANs support sparsity via the $\ell_1$ regularization. To mirror this while preserving sampling access, we can impose the $\| \vec{w} \|_1 \le 1$ regularization via coherent squaring of the real amplitudes, as described in \cref{subsec:parametrization}.

\subsection{Numerical illustration of Gaussian state preparation via QKAN}
\label{subsec:gaussian_numerics}

To complement our theoretical analysis, we provide a concise numerical example illustrating Gaussian state preparation via QKAN and highlighting several practical considerations. As discussed in \cref{sec:application_state_preparation}, the only approximation in our two-layer construction arises in the second layer, where the exponential decay $e^{-\beta(x+1)}$ is implemented using a Chebyshev polynomial. The first layer, which prepares the $D$-dimensional grid and evaluates low-degree Chebyshev polynomials, is exact and incurs no approximation or block-encoding error. Supplementary Note 4.B provides a worst-case bound on the approximation error of the exponential; here we examine how the approximation behaves numerically and how it affects the resulting state.

We specialize to $D=2$ and consider a $32\times 32$ grid corresponding to $n=5$ qubits per dimension, setting $\beta = 6$. This 2D example is also illustrated in \cref{fig:2d_gaussian}. The first QKAN layer computes
\begin{equation}
    z_{(i,j)}
    = \tfrac{1}{2}\left[T_2(x_i)+T_2(y_j)\right]
    = x_i^2 + y_j^2 - 1,
\end{equation}
exactly at each grid point, producing a diagonal operator that encodes a (shifted) squared radius on the grid.

The second layer approximates $x\mapsto e^{-\tilde\beta(x+1)}$ with $\tilde\beta = \beta/2$ using a degree-$d$ polynomial $P_d(x)$. The standard QSVT constructions require polynomials of definite parity, so we first decompose the target function into even and odd components,
\begin{equation}
    f(x) = \tfrac{1}{2}\big[f_{\mathrm{even}}(x) + f_{\mathrm{odd}}(x)\big] \;=\; \tfrac{1}{2}\big[(f(x)+f(-x)) + (f(x)-f(-x))\big],
\end{equation}
approximate each component by a truncated Taylor series of degree $d$, and combine them via an equal-superposition LCU using one additional ancilla qubit, as in \cref{theorem:qet_of_hermitian_matrices}. To satisfy the amplitude constraints $|P_d(x)| \le \tfrac{1}{2}$ needed for \cref{theorem:qet_of_hermitian_matrices}, we rescale the target function to $\tfrac{1}{2}e^{-\tilde\beta(x+1)}$; this introduces only a constant-factor overhead in the final amplitude-amplification step. Finally, to implement the resulting polynomials using QSVT, we compute the required phase angles numerically using the \texttt{pyqsp} Python package~\citep{martyn2021grand}. 

\cref{fig:qkan_gaussian_numerics}~(a) shows the resulting degree-3 polynomial $P_3(x)$ approximating $\tfrac{1}{2}e^{-\tilde\beta(x+1)}$. The approximation is accurate on the interval $[-1,1]$, as required for QKAN, and deteriorates outside this range. Applying the polynomial to the output of the first layer yields a diagonal operator whose entries approximate the target Gaussian values. \cref{fig:qkan_gaussian_numerics}~(b) displays the absolute error in the normalized two-dimensional Gaussian state prepared using this degree-3 polynomial. For each grid point $(i,j)$ we plot $\bigl|\psi_{\mathrm{exp}}(i,j) - \psi_{\mathrm{QKAN}}(i,j)\bigr|$, where $\psi_{\mathrm{QKAN}}$ is the normalized output state and $\psi_{\mathrm{exp}}$ is the ideal Gaussian. The error is largest near the center of the distribution. This matches the behavior in \cref{fig:qkan_gaussian_numerics}~(a): the center corresponds to the region where $x$ is close to $-1$ and the polynomial approximation error is maximal. 

Finally, \cref{fig:qkan_gaussian_numerics}~(c) shows the $\ell_2$ error, $\bigl\|\psi_{\mathrm{exp}} - \psi_{\mathrm{QKAN}}\bigr\|_2$, as a function of the polynomial degree $d$, together with the theoretical bound from Supplementary Note 4.B. The empirical error decays exponentially with $d$, in agreement with the lemma, and plateaus at $d=20$ when the error reaches machine precision of $10^{-14}$–$10^{-15}$.

\begin{figure}
    \centering
    \includegraphics[width=\linewidth]{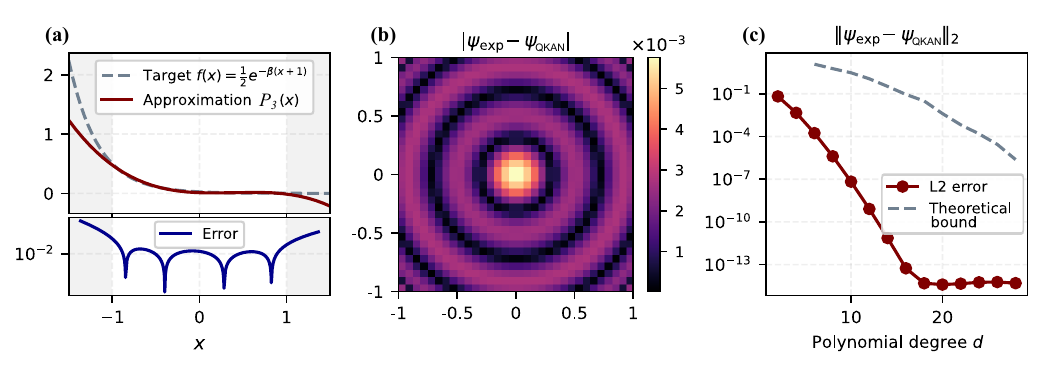}
    \caption{\textit{Numerical illustration of Gaussian state preparation via QKAN.}
    (a) Degree-3 polynomial $P_3(x)$ approximating $\tfrac{1}{2}e^{-\beta(x+1)}$ on $[-1,1]$.  
    (b) Absolute amplitude error $\bigl|\psi_{\mathrm{exp}}(i,j) - \psi_{\mathrm{QKAN}}(i,j)\bigr|$ of the normalized 2D Gaussian state prepared using $P_3$.  
    (c) $\ell_2$-error between the prepared and target states as a function of polynomial degree $d$, compared to the theoretical bound from Supplementary Note 4.B; the empirical error decreases exponentially until saturating at machine precision.}
    \label{fig:qkan_gaussian_numerics}
\end{figure}

\subsection{Generalized state preparation via CHEB-QKAN}\label{subsec:state_preparation_via_cheb_qkan}

When discussing multivariate state preparation in \cref{sec:application_state_preparation} we assumed the input register $\vec{x}$ encoded a regular $D$-dimensional grid to prepare multivariate distributions on a grid. We now generalize this to an arbitrary input by showing that a CHEB-QKAN layer implements a general multivariate state-preparation routine on any block-encoded input vector $\vec{x}$.  \cref{theorem:state_preparation_via_cheb_qkan} guarantees that we can prepare a quantum state with amplitudes matching the entries of an arbitrary $K$-dimensional CHEB-QKAN layer, as long as the input and weights can be block-encoded efficiently. The proof is deferred to Supplementary Note 5 and proceeds by applying the block-encoding to the uniform superposition followed by amplitude amplification.

\begin{theorem}[Multivariate state preparation via CHEB-QKAN]
    \label{theorem:state_preparation_via_cheb_qkan}
    Let $\varepsilon \in (0, \frac{1}{2})$. We are given access to a controlled diagonal $(1, a_x, \varepsilon_x)$-block-encoding $U_x$ of an input vector $\vec{x} \in [-1,1]^{N}$, and access to $d+1$ controlled diagonal $(1, a_w, \varepsilon_w)$-block-encodings $U_{w^{(r)}}$ of weight vectors $\vec{w}^{(r)} \in [-1,1]^{NK}$. Let $\mathcal{N}^2 := \sum_{q=1}^{K}\left(\frac{1}{N}\sum_{p=1}^{N} \phi_{pq}(x_p)\right)^2$ and $d$ be the maximal degree of Chebyshev polynomials used in parameterization of activation functions $\phi_{pq}$. If $\varepsilon_x \le \frac{\mathcal{N}^2}{144Kd^2}\varepsilon^2$ and $\varepsilon_w \le \frac{\mathcal{N}}{3\sqrt{K}}\varepsilon$, then we can prepare a $\ell_2$ normalized quantum state $\ket{\psi}$ with amplitudes corresponding to a $\emph{CHEB-QKAN}$ layer such that
    \begin{equation}
        \left\|\ket{\psi} - \frac{1}{\mathcal{N}}\sum_{q=1}^{K} \left( \frac{1}{N}\sum_{p=1}^{N} \phi_{pq}(x_p)\right)\ket{q}_k \right\|_2 \le \varepsilon,
    \end{equation}
    The procedure succeeds with arbitrarily high probability by using $\mathcal{O}(\sqrt{K}d^2/\mathcal{N})$ applications of controlled-$U_x$ and controlled-$U_{w^{(r)}}$ and their adjoint versions.
\end{theorem}

Such a strategy can be viewed as a multivariate extension of the nonlinear amplitude transformation procedure outlined by \citet{guo2024nonlinear} because the amplitudes of the prepared state are multivariate functions of the input vector. In their approach, a nonlinear amplitude transformation unitary is applied to a uniform superposition state, followed by amplitude amplification. On the other hand, achieving a multivariate version of exponential improvement through importance sampling, as proposed by~\citet{rattew2023nonlineartransformationsquantumamplitudes}, remains an open problem, due to the challenges associated with implementing importance sampling in the multivariate setting.

\section*{Acknowledgments} 
We thank Juan Carrasquilla, Naixu Guo, and Xiufan Li for discussions and feedback. 
P.-W.H. is partially supported by the Engineering and Physical Sciences Research Council (EPSRC) Doctoral Training Partnership (DTP) under grant EP/W524311/1, with a CASE Conversion Studentship in collaboration with Quantum Motion.
This work is supported by the National Research Foundation, Singapore, and A*STAR under its CQT Bridging Grant and its Quantum Engineering Programme under grant NRF2021-QEP2-02-P05. 
Portions of this manuscript were drafted or edited with the assistance of ChatGPT to improve clarity and style. 

\section*{Data Availability Statement}
This manuscript does not report data generation or analysis.

\section*{Code Availability Statement}
This manuscript does not report any new code. All algorithmic details necessary for implementation are described explicitly in the manuscript.

\section*{Competing Interests}
The authors declare no competing interests.

\section*{Author Contributions}
All authors contributed to the development of the QKAN construction and to writing the final manuscript. P.I. and P.-W.H. worked out the proofs of theorems. L.P. and P.R. conceived the project. P.R. supervised the project.

\bibliography{bibliography}
\clearpage
\appendix

\counterwithin{theorem}{section}
\counterwithin{equation}{section}
\counterwithin{definition}{section}
\crefalias{section}{appendix}

\begin{center}
\noindent{\large\bfseries \appendixname{} for ``\textsl{QKAN: Quantum Kolmogorov-Arnold Networks with Applications in Machine Learning and Multivariate State Preparation}''}
\end{center}
\setcounter{secnumdepth}{3}
\appendixtableofcontents

\section{Cloning diagonal block-encodings}\label{sec:cloning_diagonal_block_encodings}
We show the simple fact that diagonal block-encodings can be duplicated across additional registers by tensoring with the identity:
\begin{lemma}[Cloning diagonal block-encodings]
    Let $U_x$ be an $(\alpha,a,\varepsilon)$-block-encoding of $\diag(x_1,\dots,x_{N})$. Then $U_x \otimes I_{k}$ is an $(\alpha,a,\varepsilon)$-block-encoding of
    \begin{equation}
        \diag(x_1,\dots,x_{N}) \otimes I_{k} = \diag(\underbrace{x_1,\dots,x_1}_{K},\dots,\underbrace{x_{N},\dots,x_{N}}_{K}).
    \end{equation}
    \label{lemma:cloning_diagonal_block_encodings}
\end{lemma}

\begin{proof}
    By definition, $\|\diag(x_1,\dots,x_{N})-\alpha(\bra{0}_a\otimes I_{n})U_x(\ket{0}_a\otimes I_{n})\| \leq \varepsilon$. Therefore,
    \begin{align*}
        &\|\diag(x_1,\dots,x_{N})\otimes I_{k} - \alpha(\bra{0}_a\otimes I_{n})U_x\otimes I_{k}(\ket{0}_a\otimes I_{n})\| \\
        =&  \|\diag(x_1,\dots,x_{N})\otimes I_{k} - \alpha(\bra{0}_a\otimes I_{n})U_x(\ket{0}_a\otimes I_{n})\otimes I_{k}\| \\
        =& \|\diag(x_1,\dots,x_{N})-\alpha(\bra{0}_a\otimes I_{n})U_x(\ket{0}_a\otimes I_{n})\| \leq \varepsilon.
    \end{align*}
    Hence, $U_x \otimes I_{k}$ is an $(\alpha,a,\varepsilon)$-block-encoding of $\diag(x_1,\dots,x_{N}) \otimes I_{k}$, as claimed.
\end{proof}

\section{Resource analysis and error propagation}\label{subsec:resource_analysis}
In this supplementary section, we provide a detailed analysis of the number of gates and auxiliary qubits for a multilayer CHEB-QKAN. Consider a CHEB-QKAN with $L$ layers, where $N^{(l)}$ denotes the number of nodes in the $l$-th layer, and $N^{(0)}$ represents the input dimension. We assume that all activation functions in each layer are linear combinations of Chebyshev polynomials of degree at most $d$. Let $C_x^{(l)}$ represent the gate complexity of constructing a controlled diagonal block-encoding corresponding to the output of the $l$-th layer, with $C_x^{(0)}$ being the gate complexity of constructing a controlled diagonal block-encoding of the input. Additionally, let $C_w^{(l)}$ denote the gate complexity of constructing a controlled diagonal block-encoding of weights for the $l$-th layer. In \cref{tab:resource_requirements} we estimate the gate complexity of constructing a diagonal block-encoding of the output of the $l+1$-th layer in terms of the number of applications of $C_x^{(l)}$ and $C_w^{(l)}$, which depend on the block-encoding methods used for the input and weights. 

Constructing a controlled version of the obtained block-encoding, which is required for the next layer, can be achieved by controlling every gate in the block-encoding construction. This process has the same asymptotic complexity as the uncontrolled block-encoding, with an additional overhead of at most $\mathcal{O}(n_{\mathrm{total}}^2)$ one- and two-qubit gates, where $n_{\mathrm{total}}$ is the total number of qubits~\citep{nielsen_quantum_2011}. Furthermore, here, we do not explicitly consider the additional one- and two-qubit gate overhead, which is upper bounded by $\mathcal{O}(n_{\mathrm{total}}^2)$, and arises from decomposing multi-controlled operations and other steps in the construction. The reason for this omission is that we expect the dominant gate cost to come from constructing the diagonal block-encodings for inputs and weights, rather than from these additional overheads.

 \Cref{tab:resource_requirements} presents the cumulative gate complexity and qubit requirements for each step in the construction of the $(l+1)$-th KAN layer. The auxiliary qubits in \cref{tab:resource_requirements} are necessary for the block-encoding of the current transformation. Since the transformations are performed sequentially, each subsequent transformation adds to the gate and qubit cost. Consequently, the final transformation, denoted as the \textit{SUM} step in \cref{tab:resource_requirements}, represents the total cost of realizing the $l+1$ layers.

\begin{table}
    \centering
    {\setlength{\tabcolsep}{2em}
    \renewcommand{\arraystretch}{1.5}
    \begin{tabular}{cll}
        \toprule
        \textit{Step} & \textit{Gates} & \textit{Auxiliary qubits} \\
        \midrule
         \textit{DILATE}  & $C_x^{(l)}$ & $a_x^{(l)}$ \\
         \textit{CHEB}  & $dC_x^{(l)}$ & $a_x^{(l)} + 1$ \\
         \textit{MUL}  & $dC_x^{(l)} + C_w^{(l)}$ & $a_x^{(l)} + 1 + a_w^{(l)}$ \\
         \textit{LCU}  & $\frac{d^2}{2}C_x^{(l)} + dC_w^{(l)}$ & $a_x^{(l)} + 1 + a_w^{(l)} + \log_2{(d+1)}$ \\
         \textit{SUM}  & $\frac{d^2}{2}C_x^{(l)} + dC_w^{(l)}$& $a_x^{(l)} + 1 + a_w^{(l)} + \log_2{(d+1)} + \log_2{N^{(l)}}$ \\
        \bottomrule &$\underbrace{\hphantom{\frac{d^2}{2}C_x^{(l)} + dC_w^{(l)}}}_{=C_x^{(l+1)}}$ & $\underbrace{\hphantom{a_x^{(l)} + 1 + a_w^{(l)} + \log_2{(d+1)} + \log_2{N^{(l)}}}}_{=a_x^{(l+1)}}$
    \end{tabular}
    }
    \caption{\textit{Resource estimate for the $(l+1)$-th CHEB-QKAN layer.} We estimate the gate complexity $C_x^{(l+1)}$ and qubit requirements $a_x^{(l+1)}$ for constructing the $(l+1)$-th layer based on the corresponding quantities of the $l$-th layer and the cost $C_w^{(l)}$ of block-encoding the weight vectors. This recursive approach is used because each layer's block-encoding serves as the primitive building block for constructing the subsequent layer. The resource requirements are presented cumulatively, with costs increasing for each step. The final \textit{SUM} step indicates the total accumulated cost, in terms of gates and auxiliary qubits, for realizing $l+1$ layers. Auxiliary qubits are required at each step to encode the current transformation.}
    \label{tab:resource_requirements}
\end{table}

To estimate the gate and qubit complexity of an $L$-layer CHEB-QKAN in terms of $C_x^{(0)}$ and $C_w^{(l)}$, we recursively expand the gate cost and obtain the following expression for $C_x^{(L)}$:

\begin{align*}
    C_x^{(L)} &= \frac{d^2}{2}C_x^{(L-1)} + dC_w^{(L-1)} \nonumber\\
              &= \frac{d^4}{4}C_x^{(L-2)} + \frac{d^3}{2}C_w^{(L-2)} + dC_w^{(L-1)} \nonumber\\
              &= \frac{d^6}{8}C_x^{(L-3)} + \frac{d^5}{4}C_w^{(L-3)} + \frac{d^3}{2}C_w^{(L-2)} + dC_w^{(L-1)} \nonumber\\
              &\;\;\vdots \nonumber\\
              &= \left(\frac{d^2}{2}\right)^{L}C_x^{(0)} + d\sum_{l=1}^{L}\left(\frac{d^2}{2}\right)^{l-1}C_w^{(L-l)}.
\end{align*}
The total number of auxiliary qubits, $a_x^{L}$, can be expressed as:
\begin{equation}
    a_x^{(L)} = a_x^{(L-1)} + 1 + a_w^{(L-1)} + \log_2{(d+1)} + \log_2{N^{(L-1)}} = a_x^{(0)} + \sum_{l=1}^{L-1} \left(1 + a_w^{(l)} + \log_2{(d+1)} + \log_2{N^{(l-1)}}\right).
\end{equation}
In summary, a $L$-layer CHEB-QKAN requires $\mathcal{O}\left(d^{2L}/2^L\right)$ applications of controlled block-encodings of the input and weights and a number of auxiliary qubits that grow linearly with $L$.

To assess the error propagation in a multi-layer QKAN, recall that by Theorem 1 in the main text we can create a block-encoding of the single-layer vector $\Phi(\vec x)$ up to the error $4d\sqrt{\varepsilon_x} + \varepsilon_w$.  Suppose we could encode weights perfectly, i.e. $\varepsilon_w = 0$. Then, starting with the $(1,a_x,\varepsilon_x)$ block-encoding $U_x$ of the input vector $\vec x$, the error in the first few layers evolves as
\begin{equation}
    \varepsilon_x \longrightarrow 4d\sqrt{\varepsilon_x} \longrightarrow 4d\sqrt{4d\sqrt{\varepsilon_x}} \longrightarrow
    \dots
\end{equation}
Consequently, after $L\ge1$ layers, the accumulated error $\varepsilon_{\rm tot}$ in the block-encoding is
\begin{equation}
  \label{eq:error-input}
  \varepsilon_{\rm tot}
  = {\varepsilon_x}^{(1/2)^L}\,(4d)^{1 + \tfrac12 + \tfrac14 + \dots + 2^{-L}}
  = \mathcal O\bigl(d^2\,{\varepsilon_x}^{(1/2)^L}\bigr).
\end{equation}
As a result, such recursive error propagation quickly amplifies any nonzero error. To achieve a constant small error in the output, $\varepsilon_{\rm tot}=O(1)$, the input block-encoding must satisfy $\varepsilon_x = O\bigl(d^{-2^L}\bigr)$, i.e., decaying super-exponentially in $L$. The same analysis applies if we consider noiseless input ($\varepsilon_x=0$) and track weight error $\varepsilon_w$.  In the first few layers, the output error evolves as
\begin{equation}
    \varepsilon_w  \longrightarrow
    4d\sqrt{\varepsilon_w} + \varepsilon_w \longrightarrow 4d\sqrt{4d\sqrt{\varepsilon_w} + \varepsilon_w} + \varepsilon_w \longrightarrow \dots
\end{equation}
and after $L\ge1$ layers the accumulated error $\varepsilon_{\rm tot}$ can be upper-bounded by
\begin{equation}
  \label{eq:error-weight}
  \varepsilon_{\rm tot}
  \le 16d^2 \sum_{k=0}^{L-1} {\varepsilon_w}^{(1/2)^k}
  = \mathcal O\bigl(d^2\,{\varepsilon_w}^{(1/2)^L}\bigr).
\end{equation}
Similarly, to achieve $\varepsilon_{\rm tot}=O(1)$ the weight block-encoding error must satisfy $\varepsilon_w = O\bigl(d^{-2^L}\bigr)$.

These error-propagation bounds translate directly into the fault-tolerant gate-count overhead when using a discrete universal set (e.g. Clifford + $T$). Preparing an arbitrary quantum state or block-encoding a classical matrix to precision $\varepsilon$ requires $\mathcal{O}(\log (\tfrac{1}{\varepsilon}))$ $T$ gates~\citep{ross2014optimal, clader2023quantum}. From our multi-layer error analysis, the input (or weight) encoding error must satisfy $\varepsilon_{x/w} = O\bigl(d^{-2^L}\bigr)$ to keep the end-to-end error $\varepsilon_{\rm tot} = \mathcal{O}(1)$ after $L$ layers. Plugging this into the $\mathcal{O}(\log (\tfrac{1}{\varepsilon}))$ scaling gives a $T$-gate overhead of
\begin{equation}
    \log (\tfrac{1}{\varepsilon_{\rm tot}}) = \log (d^{2^L}) = 2^L \log d,
\end{equation}
i.e., an exponential dependence on the network depth $L$. Hence the total $T$ gate complexity of a $L$-layer CHEB-QKAN is
\begin{equation}
    C_{\rm tot} = \mathcal{O}\bigl((2d^2)^L \log d \bigr).
\end{equation}

\section{Output estimation of CHEB-QKAN}\label{subsec:output_estimation_cheb_qkan_appendix}
In this supplementary section, we show how amplitude estimation on a suitably constructed Hadamard-test circuit yields the precision and complexity stated in \cref{theorem:amplitude_estimation_cheb_qkan}. 

\begin{manualtheorem}{2}[Output estimation of CHEB-QKAN]
    Given access to a controlled diagonal $(1, a_x, \varepsilon_x)$-block-encoding $U_x$ of an input vector $\vec{x} \in [-1,1]^{N}$, and access to $d+1$ controlled diagonal $(1, a_w, \varepsilon_w)$-block-encodings $U_{w^{(r)}}$ of weight vectors $\vec{w}^{(r)} \in [-1,1]^{NK}$, we can estimate the value $\Phi(\vec{x})_q = \frac{1}{N}\sum_{p=1}^{N} \phi_{pq}(x_p)$ of the $q$-th component of the $\emph{CHEB-QKAN}$ layer to $\left( 4d\sqrt{\varepsilon_x} + \varepsilon_w + \delta\right)$-precision using $\mathcal{O}\left(d^{2}/\delta\right)$ applications of controlled-$U_x$ and controlled-$U_{w^{(r)}}$ and their adjoint versions. 
\end{manualtheorem}
\begin{proof}
By Theorem 1 in the main text, we can construct a diagonal block-encoding $U_{\Phi}$ such that
\begin{equation}
    \left\|\bra{0}_{\mathrm{aux}} U_{\Phi} \ket{0}_{\mathrm{aux}} - \mathrm{diag}\left( \frac{1}{N}\sum_{p=1}^{N}\phi_{p1}(x_p), \dots, \frac{1}{N}\sum_{p=1}^{N}\phi_{pK}(x_p)\right)\right\| \le 4d\sqrt{\varepsilon_x} + \varepsilon_w
\end{equation}
using $\mathcal{O}\left(d^2\right)$ applications of controlled-$U_x$ and controlled-$U_{w^{(r)}}$ and their adjoint versions. Consider $\tilde U_{\Phi} := (H\otimes I_{\mathrm{aux}+k})CU_{\Phi}(H\otimes I_{\mathrm{aux}+k})$. Applying $\tilde U_{\Phi}$ to the computational basis state $\ket{0}\ket{0}_{\mathrm{aux}}\ket{q}_k$ yields 

\begin{align*}
    \tilde U_{\Phi}\ket{0}\ket{0}_{\mathrm{aux}}\ket{q}_k &= \ket{0}\left( \frac{U_{\Phi} + I_{\mathrm{aux}+k}}{2}\right)\ket{0}_{\mathrm{aux}}\ket{q}_k + \ket{1}\left( \frac{U_{\Phi} - I_{\mathrm{aux}+k}}{2}\right) \ket{0}_{\mathrm{aux}}\ket{q}_k\\
    &= \left(\frac{\alpha_q + 1}{2} \right)\ket{0}\ket{0}_{\mathrm{aux}}\ket{q}_k + \sqrt{1-\left(\frac{\alpha_q + 1}{2} \right)^2}\ket{\perp},
\end{align*}
where $|\alpha_q - \frac{1}{N}\sum_{p=1}^{N} \phi_{pq}(x_p) | \le 4d\sqrt{\varepsilon_x} + \varepsilon_w$ by Theorem 1 in the main text and $\ket{\perp}$ is orthogonal to $\ket{0}\ket{0}_{\mathrm{aux}}\ket{q}_k$. Using amplitude estimation, we can estimate the absolute value $|\frac{\alpha_q + 1}{2}|$ of the amplitude of the good state $\ket{0}\ket{0}_{\mathrm{aux}}\ket{q}_k$ to the additive $\delta$-precision using $\mathcal{O}\left( 1/\delta \right)$ queries to $\tilde U_{\Phi}$ and $\tilde U_{\Phi}^{\dagger}$~\citep{Brassard_2002}, resulting in a total complexity of $\mathcal{O}\left(d^2/\delta\right)$ queries to controlled-$U_{x}$ and controlled-$U_{w^{(r)}}$ and their adjoint versions. Since $|\alpha_q| \le 1$, it follows that $\frac{\alpha_q + 1}{2} \ge 0$ and we can infer the value of $\alpha_q$ from the estimated amplitude. 
\end{proof}

\section{Multivariate Gaussian state preparation}

In this supplementary section, we show how to prepare a multivariate Gaussian state with the use of the grid encoding in \cref{lemma:multivariate_grid} and prove the complexity stated in \cref{theorem:multivariate_gaussian}.

\begin{manuallemma}{3}[Multivariate grid encoding]
Let $G_D \;=\;\diag\bigl(x_{(i_1,\dots,i_D)}\bigr)_{(i_1,\dots,i_D)\in\{0,\dots,2^n-1\}^D}$, where
\begin{equation}
    x_{(i_1,\dots,i_D)}
=\bigl(-1 + i_1\,s,\;-1 + i_2\,s,\;\dots,\;-1 + i_D\,s\bigr)
\;\in[-1,1]^D,
\end{equation}
be a uniform (vectorized) $D$-dimensional grid on $[-1,1]^D$ with step size $s = \tfrac{2}{2^n-1}$ in every direction. The dimension of $G_D$ is $D\,2^{nD}$. Then,
\begin{equation}
    G_1 = \sum_{i=1}^{n}\Bigl(\frac{2^{i-1}}{2^n-1}\Bigr) \, I_{i-1}\otimes XZX \otimes I_{n-i} \quad \text{and} \quad G_D \;=\; \sum_{j=1}^{D} I_{n(j-1)} \otimes G_1 \otimes I_{n(D-j)}\otimes \ketbra{j}{j},
\end{equation}
and we can create a $(1, D \lceil \log n\rceil,0)$-block-encoding of $G_D$ using $\mathcal{O}(Dn(\log n + \log D))$ two-qubit gates.
\end{manuallemma}

\begin{manualtheorem}{4}
    We can prepare a $Dn$-qubit quantum state $\ket{\psi}$ with amplitudes corresponding to a $D$-dimensional Gaussian distribution on a regular square grid of size $(2^n)^D$ such that 
    \begin{equation}
        \left\|\ket{\psi} \; - \; \frac{1}{\widetilde F_{\rm exp}}\sum_{i_1,\dots,i_D}^{2^n}\exp(-\tfrac{\beta}{2} \sum_{j=1}^D x_{i_j}^2)\,\ket{i_1,\dots,i_D} \right\|_2 \le \delta,
    \end{equation}
    where $\widetilde F_{\rm exp}$ normalizes the target state. The procedure succeeds with arbitrarily high probability by using $\mathcal{\widetilde O}\bigl( \beta^{\tfrac{D}{4}+\tfrac{1}{2}} \, n \, \log \tfrac{1}{\delta} \bigr)$ two-qubit gates and $D\lceil \log n\rceil + \lceil \log D\rceil + 4$ ancilla qubits.
\end{manualtheorem}

\subsection{Multivariate grid encoding}\label{subsec:appendix_multivariate_grid_encoding}

In the following, we extend the univariate grid encoding in \cref{lemma:univariate_grid_encoding} to $D$-dimensions to encode a multivariate grid, thus proving \cref{lemma:multivariate_grid}.

\begin{lemma}[Univariate grid encoding -- ~\citep{rosenkranz2025quantum}] 
Let $G_1 = \diag (-1 + j\,s)_{j=0,\dots,2^n-1}$ be a uniform grid on the interval $[-1,1]$ with step size $s = \tfrac{2}{2^n-1}$. Then,
\begin{equation}
    G_1 = \sum_{i=1}^{n}\Bigl(\frac{2^{i-1}}{2^n-1}\Bigr) \, I_{i-1}\otimes XZX \otimes I_{n-i}
\end{equation}
and we can create a $(1, \lceil \log n\rceil,0)$-block-encoding of $G_1$ using $\mathcal{O}(n \log n)$ two-qubit gates.
\label{lemma:univariate_grid_encoding}
\end{lemma}

\begin{proof}[Proof of \cref{lemma:multivariate_grid}]
Let $m = \lceil \log n\rceil$ and 
\begin{equation}
A \;=\;\sum_{i=1}^{n}\sqrt{\frac{2^{i-1}}{2^n-1}}\;\ketbra{i}{0}_m
\;+\;\bigl(\text{u.c.}\bigr),
\end{equation}
be a state-preparation unitary acting on $m$ ancilla qubits, where (u.c.) denotes an arbitrary unitary complement. For each $j=1,\dots,D$ define
\begin{equation}
    B_j
    =\sum_{i=1}^{n} \underbrace{I_{m(j-1)} \otimes \ketbra{i}{i}\otimes I_{m(D-j)}}_{Dm} \; \otimes \; \underbrace{I_{n(j-1)} \otimes I_{i-1} \otimes XZX \otimes I_{n-i} \otimes I_{n(D-j)}}_{Dn} \;+\;\bigl(\text{u.c.}\bigr)
\end{equation}
with underbraces denoting the numbers qubits the operators are acting on. Let $k = \lceil \log D\rceil$. Then, the unitary
\begin{equation}
    U_{G_D} =\Bigl(\underbrace{A^\dagger\otimes \dots \otimes A^\dagger}_{D \:\:\text{times}}\otimes I_{Dn}\otimes I_{k}\Bigr) \Bigl(\sum_{j=1}^D B_j\otimes\ketbra{j}{j} \;+\;\bigl(\text{u.c.}\bigr)\Bigr) \Bigl(\underbrace{A\otimes \dots \otimes A}_{D \:\:\text{times}}\otimes I_{Dn}\otimes I_{k}\Bigr) 
\end{equation}
is a $(1, D\,m,0)$–block-encoding of $G_D$ because
\begin{align*}
    \bra{0}_{Dm}\,U_{G_D}\,\ket{0}_{Dm} &= \Bigl(\bra{0}_{Dm}\,A^{\dagger\otimes D}\otimes I_{Dn}\otimes I_k\Bigr) \;\Bigl(\sum_{j=1}^D B_j\otimes\ketbra{j}{j} \;+\;\bigl(\text{u.c.}\bigr) \Bigr)\; \Bigl(A^{\otimes D}\ket{0}_{Dm}\,\otimes I_{Dn}\otimes I_k\Bigr) \\
    &=\sum_{j=1}^D I_{n(j-1)}\otimes \Bigl( \sum_{i=1}^{n} \bra{0}_m\,A^{\dagger} \ketbra{i}{i} A \ket{0}_m \: I_{i-1} \otimes XZX \otimes I_{n-i} \Bigr) \otimes I_{n(D-j)} \otimes \ketbra{j}{j} \\
    &=\sum_{j=1}^D I_{n(j-1)}\otimes \Bigl( \sum_{i=1}^{n}\Bigl(\frac{2^{i-1}}{2^n-1}\Bigr) \, I_{i-1}\otimes XZX \otimes I_{n-i} \Bigr) \otimes I_{n(D-j)} \otimes \ketbra{j}{j} \\
    &=\sum_{j=1}^D I_{n(j-1)}\otimes G_1 \otimes I_{n(D-j)} \otimes \ketbra{j}{j}.
\end{align*}
Finally, one verifies that $G_D \;=\; \sum_{j=1}^{D} I_{n(j-1)} \otimes G_1 \otimes I_{n(D-j)}\otimes \ketbra{j}{j}$ because 
\begin{align*}
    \bra{i_1, i_2, \dots,i_D}\, G_D\, \ket{i_1, i_2, \dots,i_D} &= \sum_{j=1}^D \bra{i_j}\, G_1 \,\ket{i_j} \; \ketbra{j}{j} \\
    &= \bigl(-1 + i_1\,s,\;-1 + i_2\,s,\;\dots,\;-1 + i_D\,s\bigr) \\
    &= x_{(i_1,\dots,i_D)}.
\end{align*}
The state preparation unitary $A$ can be constructed using $\mathcal{O}(n)$ two-qubit gates using a state preparation circuit with real coefficients~\citep{bergholm2005quantum}. The operator $\sum_{j=1}^D B_j\otimes\ketbra{j}{j} \;+\;\bigl(\text{u.c.}\bigr)$ is composed of $D$ $k$-qubit-controlled $B_j$ operators applied in sequence, where each $B_j$ requires $n$ $m$-qubit controlled $Z$ gates. Since an $m$-qubit-controlled $Z$ gate can be implemented using $\mathcal{O}(m)$ two-qubit gates~\citep{nie2024quantum, da2022linear}, the total two-qubit gate cost of $U_{G_D}$ is $\mathcal{O}(Dn(m+k)) = \mathcal{O}(Dn(\log n + \log D))$.
\end{proof}

\subsection{Polynomial approximation to exponential decay \texorpdfstring{$e^{-\beta(x+1)}$}{exp(-beta(x+1)}}
One can show by simple truncation the Taylor series expansion of $e^{-\beta(x+1)}$ on the interval $[-1,1]$ that $e^{-\beta(x+1)}$ can be approximated to $\varepsilon$ precision using a polynomial of degree $d = \mathcal{O}\bigl(\sqrt{\beta}\log \tfrac{1}{\varepsilon} \bigr)$. \cref{lemmas:poly_approximation_of_exp} makes this precise.

\begin{lemma}[Adapted from \citep{sachdeva2014faster}]
\label{lemmas:poly_approximation_of_exp}
$\forall  \beta>0, \varepsilon\in(0,1/2],$ there exists a polynomial $P_d$ of degree $d= \lceil\sqrt{2\lceil\max[\beta e^2,\log{(2/\varepsilon)}]\rceil\log{(4/\varepsilon)}}\rceil$ such that 
\begin{align*}
\max_{x \in [-1,1]}| P_d(x)-e^{-\beta(x+1)}| \le \varepsilon.
\end{align*}
\end{lemma}
This polynomial can be efficiently decomposed without error in the Chebyshev basis, and is hence implementable without error via QSVT with $\mathcal{O}\bigl(\sqrt{\tilde\beta}\log \tfrac{1}{\varepsilon} \bigr)$ queries to the block-encoding of $\Phi(G_D)$~\citep{childs2017quantum, low2017optimal, gilyen2019quantum, tang2024cs}.

\subsection{Bounding the filling fraction of a \texorpdfstring{$D$}{D}-dimensional Gaussian}\label{subsec:bounding_the_filling_fraction}
We demonstrate in the main text that the block-encoding of $\Phi_2(G_D)$ yields an approximation of the target $D$-dimensional Gaussian when applied to the uniform superposition state. The resulting state is subnormalized and needs to be amplified using amplitude amplification, which requires knowing a lower bound on the success probability $p$, which is given by the (discretized) $\ell_2$ filling fraction $\widetilde F_{P_d}^2$ of the approximating polynomial $P_d(x)$ on $x\in[-1,1]$: 
\begin{equation}
    p \; = \; \left\|\Phi_2(G_D)\ket{+}_{Dn}\right\|_2^2 \; = \; \frac{1}{2^{Dn}} \sum_{i_1,\dots,i_D} P_d\Bigl(\tfrac{2}{D}\sum_{j=1}^D x_{i_j}^2-1\Bigr)^2 \; = \; \widetilde F_{P_d}^2.
\end{equation}
In the following, we will bound $\widetilde F_{P_d}$ by first approximating it with its "target" version $\widetilde F_{\rm exp}$, defined as 
\begin{equation}
    \widetilde F_{P_d} \; \approx \; \widetilde F_{\rm exp} \; = \; \Bigl(\frac{1}{2^{Dn}} \sum_{i_1,\dots,i_D} \exp \bigl( -\beta \sum_{j=1}^D x_{i_j}^2\bigr)\Bigr)^{\tfrac{1}{2}}  \; = \; \Bigl(\frac{1}{2^{n}} \sum_{i=0}^{2^n-1} \exp \bigl( -\beta x_{i}^2\bigr)\Bigr)^{\tfrac{D}{2}}
\end{equation}
and then approximating $\widetilde F_{\rm exp}$ with its continuous version, for which the lower bound can be computed analytically:
\begin{equation}
    \widetilde F_{\rm exp} \; \approx \; F_{\rm exp} \;=\; \Bigr(\frac{1}{2^D}\int_{[-1,1]^D} dx_1\cdots dx_D \: \exp \bigl(-\sum_{i=1}^D \beta_i\,x_i^2\bigr) \Bigr)^{\tfrac{1}{2}} \;=\; \Bigl(\frac{1}{2} \int_{[-1,1]} dx \, \exp \bigl( -\beta x^2\bigr)\Bigr)^{\tfrac{D}{2}}.
\end{equation}
First, we consider the approximation error from the polynomial expansion:
\begin{align*}
    \left|\widetilde F_{\rm P_d}^2 - \widetilde F_{\rm exp}^2 \right| &= \left\lvert\left(\frac{1}{2^n}\right)^D\sum_{i_1,\dots,i_D} P_d\Bigl(\tfrac{2}{D}\sum_{j=1}^D x_{i_j}^2-1\Bigr)^2 - \left(\frac{1}{2^n}\right)^D \sum_{i_1,\dots,i_D} \exp \bigl( -\beta \sum_{j=1}^D x_{i_j}^2\bigr)\right\rvert \\
    &\le \left(\frac{1}{2^n}\right)^D \sum_{i_1,\dots,i_D} \left\lvert P_d\Bigl(\tfrac{2}{D}\sum_{j=1}^D x_{i_j}^2-1\Bigr)^2 - \exp \bigl( -\beta \sum_{j=1}^D x_{i_j}^2\bigr)\right\rvert.
\end{align*}
By the polynomial approximation (\cref{lemmas:poly_approximation_of_exp}), we are guaranteed that $\max_{x \in [-1,1]}| P_d(x)-e^{-\tfrac{D\beta}{4}(x+1)}| \le \varepsilon$. Consequently, $| P_d(x)^2-e^{-\tfrac{D\beta}{2}(x+1)}| \le 2\varepsilon \cdot \max_{x \in [-1,1]} \{ P_d(x), e^{-\tfrac{D\beta}{4}(x+1)}\} \le 2\varepsilon$. As a result, we obtain
\begin{equation}
    \left|\widetilde F_{\rm P_d}^2 - \widetilde F_{\rm exp}^2 \right| \le 2\varepsilon.
\end{equation}
Second, we consider the discretization error for the filling fraction, i.e., the difference between the discretized and continuous versions $\widetilde F_{\rm exp}$ and $F_{\rm exp}$. Recall the standard result on Riemann sums:
\begin{lemma}[Error bound on the Riemann sum]
    \label{lemma:error_bounds_riemann sums}
	Suppose that $f\colon [a,b]\rightarrow \mathbb{R}$ is twice continuously differentiable. 
	Let $x_i = \left( (b-a)i/N + a\right)$, then
	\begin{align*}
		\left|\frac{b-a}{N}\sum_{i=0}^{N-1} f(x_i)-\int_a^b f(x) dx\right|\leq \frac{(b-a)^2}{2N}|f'(x)|_{\mathrm{max}}^{x \in [a,b]}.
	\end{align*}
\end{lemma}
Here we consider the left Riemann sum for bounds, hence, when creating the multivariate grid, we need to replace the univariate grids to fit the grid points of the left Riemann sum. This can be obtained by rescaling the univariate grid by $\frac{2^n -1}{2^n}$ and shifting downwards by $\frac{1}{2^n}$, which can be a weighted sum of the univariate grid and identity via LCU~\citep{childs2012hamiltonian, gilyen2019quantum}.

For the one-dimensional case, the discretization error for the filling fraction can be found as follows:
\begin{equation}
    \left\lvert\frac{1}{2^n}\sum_{i=0}^{2^n-1} e^{-\beta x_i^2} - \frac{1}{2}\int_{-1}^1 e^{-\beta x^2}dx\right\rvert \le \frac{1}{2^n}\sqrt{\frac{2\beta}{e}},
\end{equation}
since the first derivative is bounded by $|(e^{-\beta x^2})'|_{\mathrm{max}} = \sqrt{\tfrac{2\beta}{e}}$. To extend the bound to the $D$-dimensional case, let $a = \frac{1}{2^n}\sum_{i=0}^{2^n-1} e^{-\beta x_i^2}$ and $b = \frac{1}{2}\int_{-1}^1 e^{-\beta x^2}dx$. Then,
\begin{equation}
    |a^D - b^D| = |a-b| \sum_{k=1}^D a^{D-k} b^{k-1} \le |a-b|\cdot D \quad \text{because} \quad a,b \le 1.
\end{equation}
Consequently, $\left| \widetilde F_{\rm exp}^2 - F_{\rm exp}^2 \right| \;=\; |a^D - b^D| \; \le \; \tfrac{D}{2^n}\sqrt{\tfrac{2\beta}{e}}$. Putting it all together, we have
\begin{equation}
    \left| \widetilde F_{P_d}^2 - F_{\rm exp}^2 \right| \;\le\; \left| \widetilde F_{P_d}^2 - \widetilde F_{\rm exp}^2 \right| + \left| \widetilde F_{\rm exp}^2 - F_{\rm exp}^2 \right| \; \le \; 2\varepsilon + \tfrac{D}{2^n}\sqrt{\tfrac{2\beta}{e}}.
\end{equation}
As a result, we have a lower bound on $\widetilde F_{P_d}$ in terms of $F_{\rm exp}$:
\begin{equation}
    \widetilde F_{P_d} \ge F_{\rm exp} \sqrt{1 - \tfrac{1}{F_{\rm exp}^2}(2\varepsilon + \tfrac{D}{2^n})} \ge F_{\rm exp} - \tfrac{1}{F_{\rm exp}}(2\varepsilon + \tfrac{D}{2^n}),
\end{equation}
assuming $2\varepsilon + \tfrac{D}{2^n} \le F_{\rm exp}^2$. Having a bound in terms of $F_{\rm exp}$ is useful because $F_{\rm exp}$ can be lower-bounded analytically:
\begin{lemma}[Continuous filling fraction of a $D$-dimensional Gaussian]
\label{lemma:continuous_filling_fraction_bound}
Let
\begin{equation}
    F_{\rm exp}^2 \;:=\;\frac{1}{2^D}\int_{[-1,1]^D} dx_1\cdots dx_D \: e^{-\sum_{i=1}^D \beta_i\,x_i^2}
\end{equation} 
be the $\ell_2$ filling fraction of a $D$-dimensional diagonal Gaussian on the interval $[-1,1]^D$. Then, for $\beta \ge 1$,
\begin{equation}
    F_{\rm exp}^2 \;\ge\;\,\Bigl(\frac{2}{3}\Bigr)^D\,\prod_{i=1}^D \beta_i^{-\tfrac12}. 
\end{equation}
\end{lemma}

\begin{proof}
The proof of the $D$-dimensional case is an extension of the one-dimensional result in Appendix E, Lemma 4 of \citep{mcardle2022quantum}. For $x \in [-1,1]$ it holds that $\exp (-\beta x^2) \ge 1 - \beta x^2$. For $\beta\ge1$, we bound $\tfrac{1}{2}\int_{-1}^1 e^{-\beta x^2}\,dx \;\ge\; \int_{-1/\sqrt\beta}^{1/\sqrt\beta} e^{-\beta x^2}\,dx \;\ge\; \int_{-1/\sqrt\beta}^{1/\sqrt\beta} (1-\beta x^2)\,dx = \tfrac{2}{3\sqrt{\beta}}$. Therefore for the $D$-dimensional product,
\begin{equation}
F_{\rm exp}^2
=\prod_{i=1}^D\frac{1}{2}\int_{-1}^1 e^{-\beta_i x_i^2}\,dx_i
\;\ge\;
\prod_{i=1}^D\frac{2}{3\sqrt{\beta_i}}
=
\Bigl(\frac{2}{3}\Bigr)^D\,\prod_{i=1}^D \beta_i^{-\tfrac12}.
\end{equation}
In particular, if $\beta_i=\beta$ for all $i$, then $F_{\rm exp} \;\ge\;\Bigl( \frac{2}{3\sqrt{\beta}}\Bigr)^{\tfrac{D}{2}}$ for $\beta \ge 1$. For $\beta < 1$, one can similarly show that $F_{\rm exp} \;\ge\;(2/3)^{\tfrac{D}{2}}$, independent of $\beta$.
\end{proof}
Using \cref{lemma:continuous_filling_fraction_bound}, we can bound
\begin{equation}
    \frac{1}{\widetilde F_{P_d}} \le \frac{2}{F_{\rm exp}} = \mathcal{O}\bigl( (\tfrac{3}{2})^{\tfrac{D}{2}} \beta^{\tfrac{D}{4}} \bigr).
\end{equation}

\subsection{Gate and qubit complexity of Gaussian state preparation}\label{subsec:gate_and_qubit_complexity_gaussian}

In this section, we calculate the two-qubit gate complexity and the number of ancilla qubits throughout our construction, thereby proving \cref{theorem:multivariate_gaussian}.

\begin{enumerate}
    \item \textit{Step 1 -- Grid encoding:} Implementing $U_{G_D}$ requires $\mathcal{O}\bigl(Dn(\log n +  \log D)\bigr)$ two-qubit gates, where $U_{G_D}$ is acting on $D(n+\lceil \log n\rceil) + \lceil \log D\rceil$ qubits, according to \cref{lemma:multivariate_grid}.
    
    \item \textit{Step 2 -- First layer:} The block-encoding of $\Phi_1(G_D)$ requires $\mathcal{O}(1)$ applications of $U_{G_D}$ and $\mathcal{O}(D\log n)$ additional two-qubit gates. It also consumes one additional ancilla qubit for the qubitization procedure \citep{lin2022lecturenotesquantumalgorithms}.
    
    \item \textit{Step 3 -- Second layer:} The block-encoding of $\Phi_2(G_D)$ requires $\mathcal{O}(d)$ applications of the block-encoding of $\Phi_1(G_D)$ and $\mathcal{O}\bigl(d\times Dn(\log n +  \log D)\bigr)$ additional two-qubit gates, and consumes two additional ancilla qubits according to Theorem 31 of Ref.~\citep{gilyen2019quantum}. Putting it all together, the two-qubit gate cost of constructing the block-encoding of $\Phi_2(G_D)$ is $\mathcal{O}\bigl( d\times(D\log n + \log D) \bigr)$. The number of main qubits $\Phi_2(G_D)$ is acting on is $Dn$ and the number of ancillas is $D\lceil \log n\rceil + \lceil \log D\rceil + 3$.

    \item \textit{Step 4 -- Amplitude amplification:} The QSVT-based fixed-point amplitude amplification requires $\mathcal{O}\bigl(\tfrac{1}{\widetilde F_{P_d}} \bigr)$ applications of the block-encoding of $\Phi_2(G_D)$ to amplify the subnormalized state $\Phi_2(G_D)\ket{+}_{Dn}$ to be arbitrarily close to its normalized version $\tfrac{1}{\widetilde F_{P_d}}\Phi_2(G_D)\ket{+}_{Dn}$. In addition, it consumes one additional ancilla and additional $\mathcal{O}\bigl(\tfrac{1}{\widetilde F_{P_d}}\times Dn(\log n +  \log D)\bigr)$ two-qubit gates for the reflection operators \citep{gilyen2019quantum, lin2022lecturenotesquantumalgorithms}. As a result, the total two-qubit gate complexity becomes $\mathcal{O}\bigl(\tfrac{d}{\widetilde F_{P_d}} \times Dn(\log n +  \log D) \bigr)$. The number of main qubits is $Dn$ and the number of ancillas is $D\lceil \log n\rceil + \lceil \log D\rceil + 4$.
\end{enumerate}
To guarantee the target accuracy $\delta$ of the state preparation, we first consider the error in the subnormalized state $\widetilde F_{P_d}\ket{\psi} = \Phi_2(G_D)\ket{+}_{Dn}$. By \cref{lemmas:poly_approximation_of_exp}, 
\begin{equation}
    \left\| \Phi_2(G_D)\ket{+}_{Dn} - \sum_{i_1,\dots,i_D}^{2^n}\exp(-\tfrac{\beta}{2} \sum_{j=1}^D x_{i_j}^2)\,\ket{i_1,\dots,i_D}\right\|_2 \le \left\| \Phi_2(G_D) - \diag \bigl( \exp(-\tfrac{\beta}{2} \sum_{j=1}^D x_{i_j}^2) \bigr) \right\|_2 \le \varepsilon.
\end{equation}
However, amplitude amplification amplifies not only the success probability but also the error in the solution. Let $\ket{\psi_{\rm target}}$ denote the target quantum state 
\begin{equation}
    \ket{\psi_{\rm target}} = \frac{1}{\widetilde F_{\rm exp}}\sum_{i_1,\dots,i_D}^{2^n}\exp(-\tfrac{\beta}{2} \sum_{j=1}^D x_{i_j}^2)\,\ket{i_1,\dots,i_D}.
\end{equation}
Then, to achieve the target approximation accuracy $ \left\|\ket{\psi} \; - \;  \ket{\psi_{\rm target}} \right\|_2 \le \delta$, note that
\begin{equation}
    \left\|\ket{\psi}  -   \ket{\psi_{\rm target}} \right\|_2 \; \le \; \left\|\ket{\psi} -  \tfrac{\widetilde F_{\rm exp}}{\widetilde F_{P_d}}\ket{\psi_{\rm target}} \right\|_2 + \left\|\tfrac{\widetilde F_{\rm exp}}{\widetilde F_{P_d}}\ket{\psi_{\rm target}} - \ket{\psi_{\rm target}} \right\|_2  \; \le \; \tfrac{\varepsilon}{\widetilde F_{P_d}} + \left\lvert \tfrac{\widetilde F_{\rm exp}}{\widetilde F_{P_d}} -1 \right\rvert.
\end{equation}
Hence, the polynomial approximation accuracy $\varepsilon$ must be chosen such that $\varepsilon + |\widetilde F_{\rm exp} - \widetilde F_{P_d}| \le \widetilde F_{P_d} \cdot \delta$. Using the bounds obtained in \cref{subsec:bounding_the_filling_fraction}, it suffices to choose $\varepsilon$ to satisfy
\begin{equation}
    \varepsilon + \tfrac{4\varepsilon}{F_{\rm exp}} \le \delta \cdot \bigl( F_{\rm exp} + \tfrac{1}{F_{\rm exp}} (2\varepsilon + \tfrac{D}{2^n}) \bigr).
\end{equation}
Rearranging for $\varepsilon$ and assuming $\delta \le \tfrac{F_{\rm exp}}{4}$, we obtain $\varepsilon \le 2\delta \bigl( F_{\rm exp} + \tfrac{D}{2^n F_{\rm exp}}\bigr)$. Furthermore, assuming $\tfrac{D}{2^n} \le F_{\rm exp}^2$, which holds for sufficiently large grid sizes, where $n\in \Omega(D\log \beta)$ it suffices to choose $\varepsilon \le 4\delta F_{\rm exp}$. Finally, plugging in the lower bound on $F_{\rm exp}$ (\cref{lemma:continuous_filling_fraction_bound}), we can definitively set $\varepsilon = 4\delta\cdot \bigl( \tfrac{2}{3\sqrt{\beta}} \bigr)^{\tfrac{D}{2}}$. This implies that the polynomial degree $\mathcal{O}\bigl(\sqrt{D\beta}\log \tfrac{1}{\varepsilon} \bigr)$ becomes a function of $D$, $\beta$, and $\delta$:
\begin{equation}
    d = \mathcal{O}\bigl(\sqrt{D\beta}\,\log \beta^{\tfrac{D}{4}} \log \tfrac{1}{\delta} \bigr).
\end{equation}
Finally, substituting the bound $\frac{1}{\widetilde F_{P_d}} = \mathcal{O}\bigl( (\tfrac{3}{2})^{\tfrac{D}{2}} \beta^{\tfrac{D}{4}} \bigr)$ from \cref{subsec:bounding_the_filling_fraction}
together with the required polynomial degree $d = \mathcal{O}\bigl(\sqrt{D\beta}\,\log \beta^{\tfrac{D}{4}} \log \tfrac{1}{\delta} \bigr)$,
we obtain the announced two-qubit-gate complexity of the $D$-dimensional Gaussian quantum state preparation as a function of $\beta,D,n$ and $\delta$:
\begin{equation}
  \mathcal{O}\bigl( \beta^{\tfrac{D}{4}+\tfrac{1}{2}} \,\log \beta \cdot (\tfrac{3}{2})^{\tfrac{D}{2}} D^{\tfrac{5}{2}}\log D \cdot n \log n \cdot \log \tfrac{1}{\delta} \bigr).
\end{equation}
Suppressing the logarithmic factors $\log n$ and $\log \beta$, and treating $D$ as a constant, we have the simplified form
\begin{equation}
     \mathcal{\widetilde O}\bigl( \beta^{\tfrac{D}{4}+\tfrac{1}{2}} \, n \, \log \tfrac{1}{\delta} \bigr).
\end{equation}
in agreement with \cref{theorem:multivariate_gaussian}.

\section{Generalized state preparation via CHEB-QKAN}\label{subsec:proof_state_prep_via_cheb_qkan}

In the following, we prove \cref{theorem:state_preparation_via_cheb_qkan} by applying the CHEB-QKAN block-encoding to the uniform superposition over the output register, and then using fixed-point amplitude amplification to boost the overlap with the target state to arbitrarily high probability. 

\begin{manualtheorem}{12}[Multivariate state preparation via CHEB-QKAN]
    Let $\varepsilon \in (0, \frac{1}{2})$. We are given access to a controlled diagonal $(1, a_x, \varepsilon_x)$-block-encoding $U_x$ of an input vector $\vec{x} \in [-1,1]^{N}$, and access to $d+1$ controlled diagonal $(1, a_w, \varepsilon_w)$-block-encodings $U_{w^{(r)}}$ of weight vectors $\vec{w}^{(r)} \in [-1,1]^{NK}$. Let $\mathcal{N}^2 := \sum_{q=1}^{K}\left(\frac{1}{N}\sum_{p=1}^{N} \phi_{pq}(x_p)\right)^2$ and $d$ be the maximal degree of Chebyshev polynomials used in parameterization of activation functions $\phi_{pq}$. If $\varepsilon_x \le \frac{\mathcal{N}^2}{144Kd^2}\varepsilon^2$ and $\varepsilon_w \le \frac{\mathcal{N}}{3\sqrt{K}}\varepsilon$, then we can prepare a $\ell_2$ normalized quantum state $\ket{\psi}$ with amplitudes corresponding to a $\emph{CHEB-QKAN}$ layer such that
    \begin{equation}
        \left\|\ket{\psi} - \frac{1}{\mathcal{N}}\sum_{q=1}^{K} \left( \frac{1}{N}\sum_{p=1}^{N} \phi_{pq}(x_p)\right)\ket{q}_k \right\|_2 \le \varepsilon,
    \end{equation}
    The procedure succeeds with arbitrarily high probability by using $\mathcal{O}(\sqrt{K}d^2/\mathcal{N})$ applications of controlled-$U_x$ and controlled-$U_{w^{(r)}}$ and their adjoint versions.
\end{manualtheorem}

\begin{proof}
By Theorem 1 in the main text, we can construct a diagonal block-encoding $U_{\Phi}$ such that
\begin{equation}
    \left\|\bra{0}_{\mathrm{aux}} U_{\Phi} \ket{0}_{\mathrm{aux}} - \mathrm{diag}\left( \frac{1}{N}\sum_{p=1}^{N}\phi_{p1}(x_p), \dots, \frac{1}{N}\sum_{p=1}^{N}\phi_{pK}(x_p)\right) \right\| \le 4d\sqrt{\varepsilon_x} + \varepsilon_w
\end{equation}
using $\mathcal{O}\left(d^2\right)$ applications of controlled-$U_x$ and controlled-$U_{w^{(r)}}$ and their adjoint versions. To prepare $\ket{\psi}$, we can apply $U_{\Phi}$ to the state $\ket{0}_{\mathrm{aux}}\ket{+}_k$ and measure the auxiliary qubits in the state $\ket{0}_{\mathrm{aux}}$. Then,
\begin{align*}
    & \left\|\frac{1}{\sqrt{p}}\bra{0}_{\mathrm{aux}} U_{\Phi} \ket{0}_{\mathrm{aux}}\ket{+}_k - \frac{1}{\sqrt{p}}\frac{1}{\sqrt{K}}\sum_{q=1}^{K} \left( \frac{1}{N}\sum_{p=1}^{N} \phi_{pq}(x_p)\right)\ket{q}_k \right\|_2 \\
    \le& \frac{1}{\sqrt{p}}\left\|\bra{0}_{\mathrm{aux}} U_{\Phi} \ket{0}_{\mathrm{aux}} - \mathrm{diag}\left( \frac{1}{N}\sum_{p=1}^{N}\phi_{p1}(x_p), \dots, \frac{1}{N}\sum_{p=1}^{N}\phi_{pK}(x_p)\right) \right\| \\
    \le& \frac{1}{\sqrt{p}}\left(4d\sqrt{\varepsilon_x} + \varepsilon_w\right),
\end{align*}
where $\sqrt{p} = \| \bra{0}_{\mathrm{aux}} U_{\Phi} \ket{0}_{\mathrm{aux}}\ket{+}_k\|_2$ is the square root of the probability of successful outcome such that $|\sqrt{p} - \frac{\mathcal{N}}{\sqrt{K}}| \le 4d\sqrt{\varepsilon_x} + \varepsilon_w$. We can boost this probability to an arbitrarily high success probability by using $\mathcal{O}(1/\sqrt{p})$ amplitude amplification steps~\citep{yoder2014fixed}.

To bound the total error within $\varepsilon \in (0, \frac{1}{2})$, let $\gamma \ge 4d\sqrt{\varepsilon_x}+\varepsilon_w$. Then we upper bound the total error by $\varepsilon$ so that $\frac{4d\sqrt{\varepsilon_x}+\varepsilon_w}{\sqrt{p}}\le \frac{\gamma}{\mathcal{N}/\sqrt{K} - \gamma} \le \varepsilon$. Rearranging the terms, we find that $\gamma \le \frac{\mathcal{N}}{\sqrt{K}}\frac{\varepsilon}{1+\varepsilon}$. Note that given $\varepsilon \le \frac{1}{2}$, $\frac{2}{3} \le \frac{1}{1+\varepsilon}$. Set $\gamma = \frac{2}{3}\frac{\mathcal{N}}{\sqrt{K}}\varepsilon$. As such, if $\varepsilon_x \le \frac{\mathcal{N}^2}{144Kd^2}\varepsilon^2$ and $\varepsilon_w \le \frac{\mathcal{N}}{3\sqrt{K}}\varepsilon$, then $4d\sqrt{\varepsilon_x} \le \frac{\mathcal{N}}{3\sqrt{K}}\varepsilon$ and the $\ell_2$ error of the prepared state is upper bounded by $\frac{4d\sqrt{\varepsilon_x}+\varepsilon_w}{\sqrt{p}}\le \varepsilon$.

Lastly, to calculate the runtime $\mathcal{O}(1/ \sqrt{p})$, we note that $\sqrt{p}$ is lower bounded by $\frac{\mathcal{N}}{\sqrt{K}} - \gamma$ which can then be lower bounded by $\frac{\mathcal{N}}{\sqrt{K}} -\frac{2}{3}\frac{\mathcal{N}}{\sqrt{K}}\varepsilon$ and furthermore, $\frac{2}{3}\frac{\mathcal{N}}{\sqrt{K}}$. Hence, the total complexity can be found to be $\mathcal{O}(\sqrt{K}d^2/\mathcal{N})$.
\end{proof}

\end{document}